\documentclass{ebarticle}
\usepackage{oamac,tikz,ebutf8}

\title{Order algebras: a quantitative model of interaction}
\author{Emmanuel Beffara}
\institute{%
  Institut de Math\'{e}matiques de Luminy \\
  UMR6206, Universit\'{e} Aix-Marseille II \& CNRS}
\date{July 7, 2011}

\begin{document}
\maketitle

\begin{abstract}
  A quantitative model of concurrent interaction is introduced.
  The basic objects are linear combinations of partial order relations, acted
  upon by a group of permutations that represents potential non-determinism in
  synchronisation.
  This algebraic structure is shown to provide faithful interpretations of
  finitary process algebras, for an extension of the standard notion of
  testing semantics, leading to a model that is both denotational (in the
  sense that the internal workings of processes are ignored) and
  non-interleaving.
  Constructions on algebras and their subspaces enjoy a good structure that
  make them (nearly) a model of differential linear logic,
  showing that the underlying approach to the representation of
  non-determinism as linear combinations is the same.
\end{abstract}

\tableofcontents

\section{Introduction} 

The theory of concurrency has developed several very different models for
interactive processes, focusing on different aspects of computation.
Among those, process calculi are an appealing framework,
because the formal language approach is well suited to
modular reasoning, allowing to study sophisticated systems by means of
abstract programming primitives for which powerful theoretical tools can be
developed.
They are also the setting of choice for extending the vast body of results
of proof theory to less sequential settings.
However, the vast majority of the semantic studies on process calculi like the
\pii-calculus have focused on the so-called interleaving operational semantics,
which is the basic definition of the dynamic of a process: the interaction of
a program with its environment is reduced to possible sequences of
transitions, thus considering that parallel composition of program components
is merely an abstraction that represents all possible ways of combining
several sequential processes into one.
In Hoare's seminal work on Communicating Sequential
Processes~\cite{hoa85:csp}, this is even an explicit design choice.

There is clearly something unsatisfactory in this state of things.
Although sophisticated theories have been established for interleaving
semantics, most of which are based on various forms of bisimulation, they
fundamentally forget the crucial (and obvious) fact that concurrent processes
are intended to model situations where some events may occur independently,
and event explicitly in parallel.
This fact is well known, and the search for non interleaving semantics for
process calculi is an active field of research, with fruitful interaction with
proof theory and denotational semantics.
Recently, the old idea of Winskel's interpretation of CCS in event
structures~\cite{win82:ccs,win87:event} has been revisited by Crafa, Varacca
and Yoshida to provide an actually non-interleaving operational semantics for
the \pii-calculus, using extensions of event structures~\cite{cvy07:event}.
Event structures are also one of the starting points of extensions of game
semantics to non-sequential frameworks, for instance in asynchronous
games~\cite{mm08:concurrent} and concurrent extensions of
ludics~\cite{fp09:orders}.
In a neighbouring line of research, the recent differential extension of
linear logic is known to be expressive enough to represent the dynamics of the
\pii-calculus~\cite{er06:diffnets,el07:pidiff}.
However the implications of this fact in the search for denotational semantics
of the \pii-calculus are still unclear, in particular the quantitative
contents of differential linear logic lacks a proper status in concurrency.

This paper presents a new semantic framework that addresses this question,
following previous work by the author~\cite{bef08:apc} on the search for
algebraically pleasant denotational semantics of process calculi.
The first step was to introduce in the \pii-calculus an additive structure (a
formal sum with zero) that represents pure non-determinism, and this technique
proved efficient enough to provide a readiness trace
semantics~\cite{oh86:spec} with a complete axiomatization of equivalence for
finite terms.
The second step presented here further extends the space of processes with
arbitrary linear combinations, giving a meaning to these combinations in terms
of quantitative testing.
This introduction of scalar coefficients was not possible in the interleaving
case, because of the combinatorial explosion that arose even when simply
composing independent traces; moving to a non-interleaving setting through
a quotient by homotopy of executions is the solution to this problem.
Growing the space of processes to get more algebraic structure is also
motivated by the idea that better structured semantics gives cleaner
mathematical foundations for the object of study, in the hope that the
obtained theory will be reusable for different purposes and that it will
benefit from existing mathematical tools.

\paragraph{Informal description}

An order algebra is defined on an \emph{arena}, which represents the set of
all observable events that may occur in the execution of a process.
Basic interaction scenarii, named \emph{plays}, are partial order relations
over finite subsets of the arena.
We then postulate two principles:
\begin{itemize}
\item 
  Linear combinations are used to represent non-determinism, which,
  although not the defining feature, is an unavoidable effect in
  concurrent interaction.
  Coefficients form the quantitative part of the model, the first thing
  they represent is how many times a given play may occur in a given
  situation.
  They can also represent more subtle things, like under which conditions a
  given play is relevant.
  This allows for the representation of features such as probabilistic
  choice, in which case coefficients will be random variables.
  In general, coefficients are taken in an arbitrary semiring with some
  additional properties.
  This use of linear combinations is a novelty of the differential λ-calculus
  and subsequent work~\cite{er03:difflambda}, although a decomposition of
  processes as formal linear combinations was first proposed by Boreale and
  Gadducci~\cite{bg03:denotational}, albeit without the quantitative aspect we
  develop here.
\item 
  The fact that some events may be indistinguishable by the environment of a
  process, typically different inputs (or outputs) on the same channel, is
  represented by a group action over the arena.
  Each element of the group acts as a permutation that represents
  a possible way of rearranging the events.
  A comparable approach was used in particular in AJM game
  semantics~\cite{ajm94:pcf,bder97:ajm} to represent the interchangeability of
  copies in the exponentials of linear logic.
\end{itemize}
Some words are borrowed from game semantics, since our objects have
similarities with games, but this is not a ``game'' semantics, at most a
degenerate one.
In particular, there is no real notion of player and opponent
interacting, since there is no polarity that could distinguish them or
distinguish inputs and outputs.
The term ``strategy'' does not really apply either
since there is no notion of choosing the next move in a given situation.
Under these circumstances, calling anything a ``game'' is kind of far fetched.

\paragraph{Outline}

Section~\ref{sec:algebras} defines order algebras from these ideas.
Arenas, plays and linear combinations of plays (simply called vectors) are
defined, with the two basic operations on vectors: \emph{synchronisation},
which extends the merging of orders to take permutations into account, and
\emph{outcome}, which is a scalar that acts as the ``result'' of a
process.
Two vectors are equivalent if they are indistinguishable by
synchronisation and outcome, and 
the order algebra is the quotient of the vectors by this equivalence.

Section~\ref{sec:logic} describes constructs involving order algebras and
their subspaces.
Cartesian and tensor products are described in terms of interaction, and the
symmetric algebra is constructed in the framework.
This  algebra is of particular interest because it represents the
basic source of non-determinism in interaction, namely the fact that any
number of interchangeable actions may occur at a given synchronisation point.

Section~\ref{sec:processes} shows how order algebras can be used to provide
fully abstract models of process calculi, with the example of the
\piI-calculus.
The crucial ingredient is a quantitative extension of the standard notion of
testing, from which the present work stems.
Standard forms of testing are obtained as particular choices of the
semiring of scalars.

\paragraph{Future work}

Order algebras as defined and studied in the present work are very finitary in
nature, because vectors are finite linear combinations of finite plays.
This setting already has an interesting structure, as this paper illustrates,
but it is unable to represent any kind of potentially infinitary behaviour.
This includes identity functions over types that are not finite dimensional,
and as a consequence we do not get a model of differential linear logic.
Handling infinity is the natural next step, and for this we need to add
topology to the structure, in order to get a sensible notion of convergence.
Order algebras will then appear not only as the quotient of combinations of
plays by equivalence, but as the separated and completed space
generated by plays.
In this line of thought, the dual space should play an important role, in
order to define duality in the logical sense.

Another direction is to exploit the fact that the semiring of scalars is a
parameter of the construction.
In particular, going from a semiring $\Scal$ to the semiring of $\Scal$-valued
random variables over a given probabilistic space properly extends the model
to a probabilistic one.
Similarly, using complex numbers and unitary transformations could provide a
way to represent quantum computation in the same framework.
Developing these ideas correctly is a line of research by itself, as the
question of denotational models for these aspects of computation is known to
be a difficult matter.

\paragraph{Related work}

Part of the construction of order algebras is concerned with modelling of
features like name binding or creation of fresh names.
The topic of proper formal handling of binders in syntax is a vast
topic known as \emph{nominal techniques} (see for instance Gabbay's
survey~\cite{gab11:nominal}), and it has been applied in particular to
construct operational semantics for process calculi in a generic
way~\cite{mp05:hdpi,cm10:names}.
We feel that our approach is orthogonal: arenas present a flattened version of
the name structure, in which remains no notion of name creation or binding (or
only indirectly); permutations are used only to relate different
\emph{occurrences} of names.
Moreover, local names, by essence, are absent from order algebras, since our
intent is to build a denotational model, in which internal behaviour is
forgotten.

Our work aims in particular at constructing models of interaction that are not
interleaving, a featured sometimes referred to as “true concurrency”.
This objective, of course, is not new, and the reference model in this respect
is that of event structures.
A relationship between our framework and event structures can be formulated:
using the simplest semiring of coefficients, namely $\{0,1\}$ with $1+1=1$
(thus losing any “quantitative” content), linear combinations of plays are
simply finite sets of plays.
The set of plays interpreting a given process turns out to be exactly the set
of configurations of the event structure interpreting this process, forgetting
any internal events.
We do not develop this correspondence in the present paper, as it is of
limited interest in the current state of development of order algebras,
however it will certainly be of great interest in the development of the
theory, notably when applying it to modelling probabilistic processes, for
which event structure semantics has been
developed~\cite{ab06:cells,vvw06:proba}.
Besides, the use of symmetry in event
structures~\cite{win07:symmetry,sw10:symmetry} has been recently identified as
a crucial feature; we defer to future work the comparison with our approach
based on group actions.

The shift from sets of configurations to formal linear combinations in the
interpretation of processes has a notable precedent in Boreale and Gadducci's
interpretation of CSP processes as formal power
series~\cite{bg03:denotational,bg06:series}, building on Rutten's work
relating coinduction and formal power series~\cite{rut99:series}.
Boreale and Gadducci's work differs from the present paper in two respects.
Firstly, their interpretation of the semiring of coefficient is of a
different nature: sum and product are seen as internal and external choice
respectively, while we interpret them as internal choice and parallel
composition without interaction.
Secondly, their technical development uses only idempotent semirings (where
$x+x=x$ for all $x$), which does not handle quantitative features, and leads
inevitably to interleaving semantics (as proved in our setting by
Theorem~\ref{thm:basis} and remarks in Section~\ref{sec:consequences}).
Nevertheless, Rutten's approach to coinduction, and the idea of coinductive
definitions by behavioural differential equations is certainly relevant to our
work and is a promising source of inspiration for the extension of the present
setting to infinitary behaviours.

\section{Order algebras} 
\label{sec:algebras}

\subsection{Arenas and plays} 

An order algebra is defined on an \emph{arena}, which represents a fixed set of
potential events.
The arena is equipped with a permutation group that represents the
non-determinism that arises when synchronising events, as described below.
Then a play is a partial order relation over a finite subset of the arena.

\begin{definition}
  An \emph{arena} $X$ is a pair $(\web{X},\pgroup[X])$ where $\web{X}$ is a
  countable set (the \emph{web} of $X$) and $\pgroup[X]$ is a subgroup of the
  group $\Perm{\web{X}}$ of permutations of $\web{X}$.
  If $\pgroup[X]$ is trivial, then $X$ is called \emph{static} and it is
  identified with its web.
\end{definition}

The points in the web are called \emph{events}, rather than moves, since there
is no actual notion of players interacting.
Permutations represent the fact that there may be several different ways for
two processes to synchronise.
In process calculus language, permutations can be seen as relating different
occurrences of the same action label.

\begin{example}\label{ex:csp}
  When modelling a simple process algebra like CSP~\cite{hoa85:csp} over an
  alphabet $A$ (with no value passing), we can use a web like $A×\Nat$, where
  $\Nat$ is the set of natural numbers; $(a,i)$ is interpreted as the $i$-th
  copy of $a$ (any other infinite set than $\Nat$ would do: the actual values
  are irrelevant).
  The permutation group will consist of all permutations of $A×\Nat$ that
  leave the first member unchanged in each pair: different occurrences of a
  given event can be freely permuted, but obviously they cannot be exchanged
  for events of a different name.
\end{example}

\begin{example}\label{ex:ccs}
  When modelling a calculus like CCS~\cite{mil89:ccs}, the same arena can be
  used as in CSP, taking for $A$ the set of action labels, including
  polarities, that is $N\uplus\set{\bar{u}}{u∈N}$ if $N$ is the set of
  names.
\end{example}

\begin{example}
  Things get more subtle when modelling a calculus with name passing like the
  $\pi$-calculus~\cite{mil99:pi}.
  For the monadic case, the arena will consist of triples $(ε,a,i)$ where $ε$
  is a polarity (input or output), $a$ is a name (either a free name or a name
  bound by an action) and $i$ is an occurrence number.
  Names bound by different input events will be considered different: in
  process terms, instead of $u(x).P\para u(x).Q$, write
  $u(x_1).P[x_1/x]\para u(x_2).P[x_2/x]$.
  The considered permutations are those that respect the name structure: if
  $σ$ maps an event $u(x_1)$ to an event $u(x_2)$, then it must map any event
  involving $x_1$ to an event of the same type involving $x_2$ instead.
  Private names, like $a$ in $\new{a}(a.P\para\bar{a}.Q)$, will not be
  represented in arenas, since by definition they cannot be involved in
  interaction with the environment, unless they are communicated by scope
  extrusion, as in $\new{a}\bar{u}a$, in which case they will be modelled
  the same way as binding input prefixes.
  This construction is developed in more detail in
  Section~\ref{sec:processes}.
\end{example}

\begin{definition}
  A \emph{play} over $X$ is a pair $s=(\web{s},\order{s})$ where $\web{s}$
  is a finite subset of $\web{X}$ (the support) and $\order{s}$ is a
  preorder over $\web{s}$; the set of plays over $X$ is written
  $\Plays{X}$.
  A play $s$ is called \emph{consistent} if the relation $\order{s}$ is a
  partial order relation (\ie\ if it is acyclic).
\end{definition}

The intuition is that a play represents a possible way a process may act:
the support contains the set of all events that will actually occur, the
preorder represents scheduling constraints for these events.
Consistency means that these constraints are not
contradictory, \ie\ that they do not lead to a deadlock.
Synchronisation, defined below, consists in combining constraints from two
plays, assuming they have the same events.
The primitive definition of plays as pre-orders is a way to make
it a total operator by separating it from the consistency condition: two plays
can synchronise even if their scheduling constraints are not compatible, but
then the result is inconsistent.

\begin{example}\label{ex:diagram}
  We will represent a (consistent) play graphically as the Hasse diagram of
  its order relation, with each node labelled by the event's name.
  By convention, when two events are part of the same orbit under
  $\pgroup[X]$, we use the same name with different indices:
  \begin{center}
    \hfill
    \tikzinline{
      \eventnode{(0,0)}{a}{left}{a}
      \eventnode{(-1,1)}{b1}{left}{b_1}
      \eventnode{(1,1)}{c}{right}{c}
      \eventnode{(1,2)}{b2}{right}{b_2}
      \eventnode{(0,3)}{d}{left}{d}
      \eventnode{(2,3)}{e}{right}{e}
      \draw (a)--(b1)--(d) (a)--(c)--(b2)--(d)  (b2)--(e);
    }
    \hfill
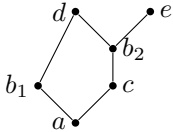
    \begin{minipage}{0.7\textwidth}
      This represents a play with support $\{a,b_1,b_2,c,d,e\}$, with the
      order relation such that $a<b_1$, $a<b_2$, $b_1<d$, $c<b_2$, $b_2<d$ and
      $b_2<e$, in an arena that has a permutation that swaps $b_1$ and $b_2$.
    \end{minipage}
  \end{center}
\end{example}

\begin{definition}
  For $r,s∈\Plays{X}$ with $\web{r}=\web{s}$, the
  \emph{synchronisation} of $r$ and $s$ is the play
  \[
    r\sync s := \bigl( \web{r}, (\order{r}\cup\order{s})^* \bigr),
  \]
  where $(⋅)^*$ denotes the reflexive transitive closure.
  Given a finite subset $A$ of $\web{X}$, define the \emph{$A$-neutral}
  play as $e_A:=(A,\operatorname{id}_A)$ where
  $\operatorname{id}_A$ is the identity relation.
\end{definition}

\begin{example}\label{ex:sync}
  We have the following synchronizations:
  \[
    \left(\!\tikzinline{
      \eventnode{(0,0)}{a1}{left}{a_1}
      \eventnode{(-0.5,1)}{b}{left}{b}
      \eventnode{(0.5,1)}{a2}{right}{a_2}
      \draw (a1)--(b) (a1)--(a2); }\!\right)
    \sync
    \left(\!\tikzinline{
      \eventnode{(0,0)}{a1}{left}{a_1}
      \eventnode{(1,0)}{b}{right}{b}
      \eventnode{(1,1)}{a2}{right}{a_2}
      \draw (b)--(a2); }\!\right)
    =
    \left(\!\tikzinline{
      \eventnode{(0,0)}{a1}{left}{a_1}
      \eventnode{(0,1)}{b}{left}{b}
      \eventnode{(0,2)}{a2}{left}{a_2}
      \draw (a1)--(b)--(a2); }\right)
  ,\quad
    \left(\!\tikzinline{
      \eventnode{(0,0)}{a1}{left}{a_1}
      \eventnode{(-0.5,1)}{b}{left}{b}
      \eventnode{(0.5,1)}{a2}{right}{a_2}
      \draw (a1)--(b) (a1)--(a2); }\!\right)
    \sync
    \left(\!\tikzinline{
      \eventnode{(0,0)}{a2}{left}{a_2}
      \eventnode{(1,0)}{b}{right}{b}
      \eventnode{(1,1)}{a1}{right}{a_1}
      \draw (b)--(a1); }\!\right)
    =
    \left(\!\tikzinline{
      \eventnode{(0,0)}{a1}{left}{a_1}
      \eventnode{(1,0)}{b}{right}{b}
      \eventnode{(1,1)}{a2}{right}{a_2}
      \draw (b) to[bend left] (a1) to[bend left] (b) -- cycle;
      \draw (b)--(a2); }\!\right)
    .
  \]
  The second one leads to an inconsistent play, since the union of the order
  relations is cyclic.
\end{example}

Note that synchronisation is a very restrictive operator because it requires
the event sets to be equal.
The possibility of synchronising on some events while keeping the
others independent, which is a natural notion, will be defined in
Section~\ref{sec:partial} using this primitive form of total synchronisation.

Commutativity of $\sync$ is immediate from the definition.
Associativity is also clear: for $r,s,t∈\Plays{X}$, $(r\sync s)\sync t$ and
$r\sync(s\sync t)$ are defined if and only if $\web{r},\web{s},\web{t}$ are
equal, and in this case we have $\order{(r\sync s)\sync t}=
(\order{r}\cup\order{s}\cup\order{t})^*=\order{r\sync(s\sync t)}$.
Because of the constraint on supports, there cannot be a neutral element.
However, among plays of a given support $A$, the neutral play $e_A$ is
actually neutral for synchronisation.

We now define the action of the permutation group $\pgroup[X]$ over the set of
plays.
Since there is usually no ambiguity, we overload the notation for group
actions: given $σ∈\pgroup[X]$, for $x∈X$ we write $σx$ for the image of $x$,
for $A⊆X$ we write $σA$ for the set of images $\set{σx}{x∈A}$, and similarly
for $r∈\Plays{X}$ we write $σr$ for the play $r$ permuted by $σ$, as defined
below.

\begin{definition}
  Let $X$ be an arena.
  The action of a permutation $σ∈\pgroup[X]$ on a play $r∈\Plays{X}$ is
  defined as
  \[
    σ r := \bigl(
      σ \web{r}, \; \set{(σ x, σ y)}{(x,y)∈\order{r}}
  \bigr) .
  \]
  The \emph{orbit} of a play $r$ in $\Plays{X}$ is the set
  \[
    \porbit[X]{r} := \set{σ r}{σ∈\pgroup[X]} .
  \]
\end{definition}

We refer the reader to some reference textbook (for instance Lang's
\emph{Algebra}~\cite{lang65:algebra}) for details on the standard
group-theoretic notions in use here.
For reference, given a group $G$ acting on a set $X$,
the \emph{stabilizer} of a point $x∈X$ in the action of $G$ is, by
definition, the subgroup of $G$ consisting of all the $σ∈G$ that
leave $x$ unchanged, \ie\ $σx=x$.
The \emph{pointwise} stabilizer of a set $A⊆X$ is the subgroup of those
that leave each point in $A$ unchanged, as opposed to the \emph{setwise}
stabilizer which includes all permutations that leave the set $A$
unchanged as a whole (\ie\ $\set{σx}{x∈A}=A$).
The index of a subgroup $H$ in a group $G$, written $(H:G)$, is the number of
left cosets of $H$ in $G$, that is the cardinal of $\set{σH}{σ∈G}$.
When $H$ is a normal subgroup of $G$, the index $(H:G)$ is the cardinal of the
quotient group $G/H$.

\begin{definition}
  Let $X$ be an arena and $r$ be a play in $\Plays{X}$.
  Let $\pstab[X]{r}$ be the stabilizer of $r$ in the action of $\pgroup[X]$
  over $\Plays{X}$ and let $\pstab[X]{\web{r}}$ be the pointwise stabilizer of
  $\web{r}$ in $\pgroup[X]$, then the \emph{multiplicity} of $r$ in $X$ is the
  index of $\pstab[X]{\web{r}}$ in $\pstab[X]{r}$:
  \[
    \multiplicity[X]{r} := (\pstab[X]{r}:\pstab[X]{\web{r}}) .
  \]
\end{definition}

Hence, the multiplicity of $r$ is the number of different ways one can permute
$r$ into itself.
Indeed, the definition as $(\pstab{r}:\pstab{\web{r}})$ exactly means
the number of permutations of $r$ into itself (elements of $\pstab{r}$), up
to permutations that leave each point of $\web{r}$ invariant (elements of
$\pstab{\web{r}}$), in other words $\multiplicity{r}$ is the order of the
group $\set[1]{\restr{σ}{\web{r}}}{σ∈\pgroup,σ r=r}$,
which is always finite since the support $\web{r}$ is finite.

\begin{example}
  Using the same conventions as in Example~\ref{ex:diagram}, we have
  \[
    \multiplicity{\tikzinline{
      \eventnode{(0,0)}{a}{left}{a}
      \eventnode{(-1,1)}{b1}{left}{b_1}
      \eventnode{(-1,2)}{c1}{left}{c_1}
      \eventnode{(1,1)}{b2}{right}{b_2}
      \eventnode{(1,2)}{c2}{right}{c_2}
      \draw (a)--(b1)--(c1) (a)--(b2)--(c2);
    }} = 2
    \qquad\text{and}\qquad
    \multiplicity{\tikzinline{
      \eventnode{(0,0)}{a}{left}{a}
      \eventnode{(-1,1)}{b1}{left}{b_1}
      \eventnode{(-1,2)}{c1}{left}{c_1}
      \eventnode{(1,1)}{b2}{right}{b_2}
      \eventnode{(0,3)}{c2}{right}{c_2}
      \draw (a)--(b1)--(c1)--(c2) (a)--(b2)--(c2);
    }} = 1
  \]
  Both plays have the same support, there are $4$ permutations of this
  support: $b_1$ and $b_2$ can be exchanged, idem for $c_1$ and $c_2$.
  In the first case if we exchange $b_1$ with $b_2$ and $c_1$ with
  $c_2$, we get the same play (permuting the $b$ but not the $c$ yields a
  different play).
  In the second case, no permutation can yield the same play.
\end{example}

\subsection{Linear combinations} 

The set of plays of an arena $X$ is independent of the group $\pgroup$,
but our idea is that plays that are permutations of each other should be
considered equivalent, since permutations exchange occurrences of
indistinguishable actions.
In the presence of permutations, however, there are several ways to
synchronise two plays, so in order to extend the definition of synchronisation
we have to be able to consider combinations of possible plays.
For genericity, and because our aim is to get a quantitative account of
interaction, we will use linear combinations, with coefficients in an
unspecified commutative semiring.

\begin{definition}
  A \emph{commutative semiring} $\Scal$ is a tuple $(\Scal,{+},{⋅},0,1)$ such
  that $(\Scal,{+},0)$ and $(\Scal,{⋅},1)$ are commutative monoids and for all
  $x,y,z∈\Scal$ it holds that $x⋅(y+z)=x⋅y+x⋅z$ and $x⋅0=0$.
  A \emph{semimodule} over $\Scal$ is a commutative monoid $(M,{+},0)$
  with an action $(⋅):\Scal×M→M$ that
  commutes with addition on both sides and satisfies $λ⋅(μ⋅x)=(λ⋅μ)⋅x$ for all
  $λ,μ∈\Scal$ and $x∈M$.
  A \emph{commutative semialgebra} over $\Scal$ is a semimodule $M$ with a
  bilinear operation that is associative and commutative.
\end{definition}

Terminology about semirings, semimodules and semialgebras is not standard, in
particular some definitions do not require both neutrals.
Sometimes, the neutrals are not required to be distinct (they are equal if and
only if the semiring is a singleton, but this is a degenerate case that we
will not consider).
In the above definitions, if all elements of $\Scal$ have additive inverses,
then $\Scal$ is a (commutative unitary) ring, and the semimodules and
semialgebras are actually modules and algebras (indeed, the action of $\Scal$
imposes the existence of additive inverses in them too).
If $\Scal$ is a field, we get the usual notions of vector space and algebra.

\begin{definition}
  Let $\Scal$ be a commutative semiring.
  The \emph{integers} of $\Scal$ are the finite sums of $1$ including the
  empty sum $0$, the \emph{non-zero integers} are the finite non-empty sums.
  We call $\Scal$ \emph{regular} if
  for every non-zero integer $n\in\Scal$,
  for all $x,y\in\Scal$, $nx=ny$ implies $x=y$.
  We call $\Scal$ \emph{rational} if every non-zero integer has a
  multiplicative inverse.
\end{definition}

In particular, regularity applied to $x=1$ and $y=0$ imposes that no non-empty
sum of $1$ can be equal to $0$, in other words $\Scal$ has characteristic
zero.
The rationality condition means that it is possible to divide by non-zero
natural numbers, or in more abstract terms that the considered semiring is a
semimodule over the semiring of non-negative rationals.
This obviously implies regularity.

Two important cases of rational semirings will be considered here.
The first case is that of commutative algebras over the field $\Rat$ of
rational numbers, which includes fields of characteristic zero (among which
rational, real and complex numbers) and commutative algebras over them.
The second case is when addition is idempotent, which includes so-called
\emph{tropical} semirings~\cite{pin94:tropical}, and the canonical examples
are that of min-plus and max-plus semirings.
Boolean algebras with disjunction as sum and conjunction as product are
another typical example.
In this case, all integers except $0$ are equal to $1$, so they obviously
have multiplicative inverses.

Throughout this paper, unless explicitly stated otherwise, $\Scal$ is any
semiring.
Note that the semiring $\Nat$ of natural numbers and the ring $\Int$ of
integers are regular but not rational.
Indeed, when using natural numbers as scalars, some properties of order
algebras will be lost, for instance the existence of bases.
Hence some statements will explicitly require $\Scal$ to be regular or
rational.

For an arbitrary set $X$, a formal linear combination over $X$ is a function
from $X$ to $\Scal$ that has a value other than $0$ on a finite number of
points.
Formal linear combinations over $X$, with sum and scalar product defined
pointwise, form the free $\Scal$-semimodule over $X$
and an element $x∈X$ is identified with its “characteristic” function
$δ_x:X→\Scal$, such that $δ_x(x)=1$ and $δ_x(y)=0$ for all $y≠x$.
If $X$ is finite, then the set of formal linear combinations is the
$\Scal$-semimodule $\Scal^X$.
For an arbitrary subset $A$ of a $\Scal$-semimodule $E$, we denote by
$\Vect[\Scal]{A}$, or simply $\Vect{A}$, the submodule of $E$ generated by
$A$, \ie\ the smallest submodule of $E$ that contains $A$, that is the set
of finite linear combinations of elements of $A$.

\begin{definition}
  Let $X$ be an arena.
  The \emph{preliminary order algebra} $\PreOrdAlg[\Scal]{X}$ is the
  free $\Scal$-semimodule over $\Plays{X}$.
  The \emph{outcome} is the linear form $\outcome{⋅}$ over
  $\PreOrdAlg[\Scal]{X}$ such that $\outcome{r}=1$ when $r$ is consistent and
  $\outcome{r}=0$ otherwise.
\end{definition}

We usually keep the semiring $\Scal$ implicit in our notations.
Vectors in $\PreOrdAlg{X}$ are finite linear combination of plays in $X$,
they represent the collection of all possible behaviours of a finite process.
The coefficients can be understood as the amount of each behaviour that
is present in the process.
Examples in further sections also illustrate that $\Scal$ can be chosen to
represent conditions on the availability of each behaviour.
The outcome represents how relevant each play is, and by the intuition exposed
in the previous section, plays with cyclic dependencies cannot happen, so they
are considered irrelevant.

\begin{example}
  Consider the CSP term $P=a{\to}(b\mathbin{\|}c)\mid a{\to}c$
  (remember that in CSP $\mid$ is the choice operator, and $\|$ is parallel
  composition).
  An interpretation of $P$ in a preliminary order algebra containing only
  $a,b,c$ as events could be
  \[
    \tikzformula{ } +
    2 \tikzformula{
      \eventnode{(0,0)}{a}{right}{a}} +
    \tikzformula{
      \eventnode{(0,0)}{a}{right}{a}
      \eventnode{(0,1)}{b}{right}{b}
      \draw (a)--(b); } +
    2 \tikzformula{
      \eventnode{(0,0)}{a}{right}{a}
      \eventnode{(0,1)}{c}{right}{c}
      \draw (a)--(c); } +
    \tikzformula{
      \eventnode{(0,0)}{a}{right}{a}
      \eventnode{(-0.5,1)}{b}{left}{b}
      \eventnode{(0.5,1)}{c}{right}{c}
      \draw (a)--(b) (a)--(c); }
  \]
  where we have a summand for each partial run of $P$.
  The coefficient $2$ in the second and fourth summands represent the fact
  that there are two ways to perform only $a$, and two ways to perform $a$
  then $c$, depending on the choice one has done.
\end{example}

\begin{definition}\label{def:psync}
  Let $X$ be an arena.
  Permuted synchronisation in $X$ is the bilinear operator
  $\psync$ over $\PreOrdAlg{X}$ such that for all plays $r,s∈\Plays{X}$,
  \[
    r \psync s := \multiplicity[X]{s}
    \sum_{\substack{s'∈\porbit[X]{s} \\ \web[1]{s'}=\web{r}}} r\sync s'
  \]
\end{definition}

Observe that this sum is always finite.
The reason is that each $s'$ can be written $σ s$ for some
$σ∈\pgroup[X]$, and $\web{σs}=\web{r}$ implies
$σ\web{s}=\web{r}$.
Since $σ s$ is determined by the action of $σ$ on $\web{s}$, there
is at most one image of $s$ for each bijection between $\web{s}$ and
$\web{r}$.
Since $\web{r}$ and $\web{s}$ are finite, the number of such bijections is
finite.

\begin{example}\label{ex:psync}
  Considering again the plays in Example~\ref{ex:sync}, we have
  \[
    \tikzformula{
      \eventnode{(0,0)}{a1}{left}{a_1}
      \eventnode{(-0.5,1)}{b}{left}{b}
      \eventnode{(0.5,1)}{a2}{right}{a_2}
      \draw (a1)--(b) (a1)--(a2); }
    \psync
    \tikzformula{
      \eventnode{(0,0)}{a2}{left}{a_2}
      \eventnode{(1,0)}{b}{right}{b}
      \eventnode{(1,1)}{a1}{right}{a_1}
      \draw (b)--(a1); }
    =
    \tikzformula{
      \eventnode{(0,0)}{a1}{left}{a_1}
      \eventnode{(0,1)}{b}{left}{b}
      \eventnode{(0,2)}{a2}{left}{a_2}
      \draw (a1)--(b)--(a2); }
    +
    \tikzformula{
      \eventnode{(0,0)}{a1}{left}{a_1}
      \eventnode{(1,0)}{b}{right}{b}
      \eventnode{(1,1)}{a2}{right}{a_2}
      \draw (b) to[bend left] (a1) to[bend left] (b) -- cycle;
      \draw (b)--(a2); }
    .
  \]
  The first term in the sum corresponds to the identity permutation, the
  second one exchanges $a_1$ and $a_2$.
  Here, all plays involved have multiplicity $1$.
\end{example}

\begin{remark}
  Permuted synchronisation is similar to the parallel composition operator of
  CSP, in the case of processes defined on the same alphabet: in $r\psync s$,
  every event of $r$ must be synchronized with some event of the same name in
  $s$.
  There is a difference between plays and processes, however, in that in a
  play, every event must occur, whereas in a process, an action may be
  cancelled, for lack of a partner action to synchronize with.
\end{remark}

Any partial function $f:\Plays{X}^n\to\Plays{X}$ extends as an $n$-linear
operator $\bar{f}:\PreOrdAlg{X}^n\to\PreOrdAlg{X}$, by setting
$\bar{f}(r_1,\ldots,r_n)=0$ when $f(r_1,\ldots,r_n)$ is undefined.
This applies in particular to synchronisation, which yields a bilinear
operator $\bar{\sync}$ over $\PreOrdAlg{X}$.
As a slight abuse of notations, we will write it simply as $\sync$ when there
is no ambiguity.

Using this convention, permuted synchronisation can be seen as a
generalisation of non-permuted synchronisation, since when the permutation
group is trivial, all multiplicities are $1$ and all orbits are singletons.
Although $\psync$ is a generalisation of $\sync$, we still use different
notations, since both operators are of interest in a given non-static arena.
The non-permuted version will be referred to as \emph{static synchronisation}
to avoid confusion.

\begin{definition}
  Let $X$ be an arena.
  Observational equivalence in $\PreOrdAlg[\Scal]{X}$ is defined as
  $u\obseqv[X]u'$ when $\outcome{u\psync v}=\outcome{u'\psync v}$ for all
  $v∈\PreOrdAlg[\Scal]{X}$.
  The \emph{order algebra} over $X$ is
  $\OrdAlg[\Scal]{X}:=\PreOrdAlg[\Scal]{X}/{\obseqv[X]}$.
\end{definition}

The scalar $\outcome{u\psync v}$ is understood as the result of testing a
process $u$ against a process $v$.
It linearly extends the basic case of single plays: $\outcome{r\sync s}$
is $1$ if $r$ and $s$ are compatible and $0$ otherwise; $\outcome{r\psync s}$
is the number of different ways $r$ and $s$ can be permuted so that they
become compatible.
Hence the definition: $u\obseqv v$ if $u$ and $v$ are indistinguishable by this
testing protocol.

\begin{example}
  Any inconsistent play is observationally equivalent to $0$, since
  synchronising it with any order yields an inconsistent play.
  Hence the synchronisation of Example~\ref{ex:psync} implies
  \[
    \tikzformula{
      \eventnode{(0,0)}{a1}{left}{a_1}
      \eventnode{(-0.5,1)}{b}{left}{b}
      \eventnode{(0.5,1)}{a2}{right}{a_2}
      \draw (a1)--(b) (a1)--(a2); }
    \psync
    \tikzformula{
      \eventnode{(0,0)}{a2}{left}{a_2}
      \eventnode{(1,0)}{b}{right}{b}
      \eventnode{(1,1)}{a1}{right}{a_1}
      \draw (b)--(a1); }
    \obseqv
    \tikzformula{
      \eventnode{(0,0)}{a1}{left}{a_1}
      \eventnode{(0,1)}{b}{left}{b}
      \eventnode{(0,2)}{a2}{left}{a_2}
      \draw (a1)--(b)--(a2); }
    .
  \]
\end{example}

\begin{lemma}\label{lemma:perm-eq}
  Let $X$ be an arena.
  For all $u∈\PreOrdAlg{X}$ and $σ∈\pgroup[X]$,
  we have $u\obseqvσ u$.
\end{lemma}
\begin{proof}
  Since plays generate the module $\PreOrdAlg{X}$, clearly $u\obseqvσ u$
  if and only if $\outcome{u\psync s}=\outcome{σ u\psync s}$ for all
  play~$s$.
  Since $u$ is a finite linear combination of plays, the definition of
  synchronisation on plays extends as $\outcome{u\psync s}
  =\multiplicity{s}\sum_{s'∈\porbit{s}}\outcome{u\sync s'}$.
  It is clear that for all $σ∈G$ we have
  $σ(u\sync s')=σ u\syncσ s'$, moreover
  outcomes are preserved by permutations, so we have
  $\outcome{u\sync s'}=\outcome{σu\syncσs'}$ for all $s'$, hence
  $\outcome{u\psync s}
  =\multiplicity{s}\sum_{s'∈\porbit{s}}\outcome{σ u\syncσ s'}$.
  Since $σ$ acts as a permutation on the orbit $\porbit{s}$,
  $σ s'$ and $s'$ range over the same set, so we have
  $\outcome{u\psync s}=\outcome{σ u\psync s}$, and finally
  $u\obseqvσ u$.
\end{proof}

  The fact that observational equivalence is preserved by linear combinations
  is immediate from the definition, since synchronisation and outcome are
  linear.
  As a consequence, in each orbit $\porbit{s}$, we can choose a representant
  $\repr{s}$ such that each vector in $\PreOrdAlg{X}$ is equivalent to a
  linear combination of representants.

\begin{definition}\label{def:choice}
  Let $X$ be an arena.
  A \emph{choice of representants} for $X$ is a pair of an idempotent map
  $A\mapsto\repr{A}$ over $\PartsFin{\web{X}}$ and an idempotent map
  $r\mapsto\repr{r}$ over $\Plays{X}$ such that
  for all $r∈\Plays{X}$ $\web{\repr{r}}=\repr{\web{r}}$, and
  for all $r,s∈\Plays{X}$, $\repr{r}=\repr{s}$ if and only if $r=σ s$
  for some $σ∈\pgroup[X]$.
\end{definition}

  So a choice of representants picks one play in each orbit under $\pgroup[X]$
  in such a way that representants have the same support if it is possible.
  There always exists such choices, and in the sequel we assume that each
  arena comes with a particular choice, written $r\mapsto\repr[X]{r}$.
  The choice function over $\Plays{X}$ induces a projection in
  $\PreOrdAlg{X}$ by linearity, and for all $u∈\PreOrdAlg{X}$ we have
  $u\obseqv\repr[X]{u}$ by Lemma~\ref{lemma:perm-eq}.

\begin{definition}\label{def:saturate}
  Let $X$ be an arena.
  \emph{Saturation} in $X$ is the linear map
  $\saturate[X]:\PreOrdAlg{X}\to\PreOrdAlg{X}$ such that for each play~$r$, 
  \[
    \saturate[X] r := \sum_{σ∈\rpgroup{\web{r}}} σ r
    \qquad\text{where}\qquad
    \rpgroup{A} := \set[2]{\restr{σ}{A}}{σ∈\pgroup[X],σ A=A}.
  \]
\end{definition}

  The set $\rpgroup{A}$ is the group of permutations of $A$ induced by $\pgroup$,
  it is isomorphic to the quotient $\pstab{\{A\}}/\pstab{A}$ where
  $\pstab{\{A\}}$ is the setwise stabilizer of $A$ in $\pgroup$ and
  $\pstab{A}$ is its pointwise stabilizer (it is easy to check that the latter
  is a normal subgroup of the former).

\begin{example}
  Again using the conventions of Example~\ref{ex:diagram}, we have
  \[
    \saturate
    \tikzformula{
      \eventnode{(0,0)}{a}{left}{a}
      \eventnode{(-1,1)}{b}{left}{b}
      \eventnode{(-1,2)}{c1}{left}{c_1}
      \eventnode{(0,1)}{c2}{above}{c_2}
      \eventnode{(1,1)}{c3}{above}{c_3}
      \draw (a)--(b)--(c1) (a)--(c2) (a)--(c3); }
    = 2 \tikzformula{
      \eventnode{(0,0)}{a}{left}{a}
      \eventnode{(-1,1)}{b}{left}{b}
      \eventnode{(-1,2)}{c1}{left}{c_1}
      \eventnode{(0,1)}{c2}{above}{c_2}
      \eventnode{(1,1)}{c3}{above}{c_3}
      \draw (a)--(b)--(c1) (a)--(c2) (a)--(c3); }
    + 2 \tikzformula{
      \eventnode{(0,0)}{a}{left}{a}
      \eventnode{(-1,1)}{b}{left}{b}
      \eventnode{(-1,2)}{c1}{left}{c_2}
      \eventnode{(0,1)}{c2}{above}{c_1}
      \eventnode{(1,1)}{c3}{above}{c_3}
      \draw (a)--(b)--(c1) (a)--(c2) (a)--(c3); }
    + 2 \tikzformula{
      \eventnode{(0,0)}{a}{left}{a}
      \eventnode{(-1,1)}{b}{left}{b}
      \eventnode{(-1,2)}{c1}{left}{c_3}
      \eventnode{(0,1)}{c2}{above}{c_1}
      \eventnode{(1,1)}{c3}{above}{c_2}
      \draw (a)--(b)--(c1) (a)--(c2) (a)--(c3); } ,
  \]
  where the factor $2$ comes from the fact that exchanging $c_2$ and $c_3$ in
  the original play does not change it.
  Indeed, the multiplicity of this play is $2$.
\end{example}

\begin{lemma}\label{lemma:saturate}
  Let $X$ be an arena.
  For all $r,s\in\Plays{X}$ such that $\web{r}=\web{s}$ we have
  $s\psync r=s\mathbin{\bar{\sync}}\saturate r$.
\end{lemma}
\begin{proof}
  We use the notations of Definition~\ref{def:saturate}.
  As $σ$ ranges over $\rpgroup{\web{r}}$, $σ r$ ranges over all the
  elements of the orbit of $r$ under $\pgroup$ that have the same support
  as $r$.
  Moreover, each play in the orbit is hit a number of times equal to
  $\multiplicity{r}$, so for any play $s$ we have the expected equality.
\end{proof}

In particular, this implies the equivalence
$u\psync v\obseqv\repr{u}\mathbin{\bar{\sync}}\saturate\repr{v}$
for all $u$ and $v$.
We will use this fact in Proposition~\ref{prop:represent} to get a
representation of arbitrary order algebras in static ones.

\begin{proposition}
  Permuted synchronisation is compatible with observational equivalence.
  Up to observational equivalence, it is associative and commutative.
\end{proposition}
\begin{proof}
  We first prove that permuted synchronisation is strictly associative.
  Consider three plays $r,s,t$.
  For all $σ∈\pgroup[X]$ we have $r\psyncσ s=r\psync s$, so we can
  assume that $r,s,t$ have equal support (if no permutation can let them have
  the same support, then synchronisation in any order is zero).
  Then we have $r\psync(s\psync t)=r\sync\saturate(s\sync\saturate t)$ by
  Lemma~\ref{lemma:saturate}, hence
  \[
    r\psync(s\psync t)
    = \sum_{σ,\tau∈\rpgroup{\web{r}}} r\syncσ(s\sync\tau t)
    = \sum_{σ,\tau∈\rpgroup{\web{r}}} (r\syncσ s)\syncσ\tau t
    = \sum_{σ,\tau∈\rpgroup{\web{r}}} (r\syncσ s)\sync\tau t
    = (r\psync s)\psync t
  \]
  using associativity of strict synchronisation and the fact that, for a fixed
  $σ$, the permutations $σ\tau$ and $\tau$ range over the same set.
  This extends to all vectors by linearity.

  Associativity implies compatibility with
  observational equivalence: for two equivalent vectors $u\obseqv u'$ and an
  arbitrary $v∈\PreOrdAlg{X}$, for all $w∈\PreOrdAlg{X}$ we have
  $\outcome{(u\psync v)\psync w}=\outcome{u\psync(v\psync w)}
  =\outcome{u'\psync(v\psync w)}=\outcome{(u'\psync v)\psync w}$ so
  $u\psync v\obseqv u'\psync v$.

  Since observational equivalence is preserved by permutations, using
  commutativity of strict synchronisation we have
  \[
    s\psync r
    = \sum_{σ∈\rpgroup{\web{r}}} s\syncσ r
    = \sum_{σ∈\rpgroup{\web{r}}} σ(σ^{-1}s\sync r)
    \obseqv \sum_{σ∈\rpgroup{\web{r}}} σ^{-1}s\sync r
    = \sum_{σ∈\rpgroup{\web{r}}} r\syncσ^{-1}s
    = r\psync s .
  \]
  which proves commutativity of permuted synchronisation up to observational
  equivalence.
\end{proof}

\begin{lemma}\label{lemma:partial-neutral}
  Let $X$ be an arena and let $(u_i)_{i∈I}$ be a finite family of
  vectors in $\PreOrdAlg{X}$.
  There exists a vector $e$ and an integer $n>0$
  such that for all $i∈I$, $u_i\psync e=n\,u_i$.
\end{lemma}
\begin{proof}
  Observe that for all finite subset $A$ of $X$, every setwise stabilizer of
  $A$ is a stabilizer of the $A$-neutral play $e_A$, so we have
  $\saturate e_A=\multiplicity{e_A}e_A$, and subsequently for all play
  $s∈\Plays{X}$ such that $\repr{\web{s}}=\repr{A}$ we get
  $s\psync e_A=\multiplicity{e_A}(s\sync e_A)=\multiplicity{e_A}s$.
  Call $P$ the set of all $\repr{\web{r}}$ such that the play $r$ has a
  non-zero coefficient in some $u_i$, then $P$ is finite since each $u_i$ is a
  finite linear combination of plays.
  Let $n$ be the least common multiple of the $\multiplicity{e_A}$ for $A$ in
  $P$, then the vector
  \[
    e := \sum_{A∈P} \frac{n}{\multiplicity{e_A}} e_A
  \]
  satisfies $e\psync u_i=nu_i$ for each $i$ by construction.
  Note that the coefficients $n/\multiplicity{e_A}$ are all natural numbers,
  by construction.
\end{proof}

\begin{corollary}
  If $\Scal$ is regular then outcome is preserved by observational
  equivalence.
\end{corollary}
\begin{proof}
  Let $u\obseqv v$ be a pair of equivalent vectors in $\PreOrdAlg{X}$.
  By Lemma~\ref{lemma:partial-neutral} there is a vector $e$ and an integer
  $n≠0$ such that  $u\psync e=nu$ and $v\psync e=nv$, so we have
  $n\outcome{u}=\outcome{u\psync e}=\outcome{v\psync e}=n\outcome{v}$.
  By regularity, we can deduce $\outcome{u}=\outcome{v}$.
\end{proof}

As a consequence, the order algebra $\OrdAlg{X}$, which is defined as the
quotient of $\PreOrdAlg{X}$ by observational equivalence, is a commutative
semialgebra over $\Scal$ with synchronisation $\psync$ as the product,
and outcome $\outcome{⋅}$ is a linear form over it.
The choice of representants for orbits of finite sets and plays induces the
following representation property of $\OrdAlg{X}$ in the static algebra
$\OrdAlg{\web{X}}$.

\begin{proposition}\label{prop:represent}
  Let $X$ be an arena.
  Define the linear map $\psplit[X]:\PreOrdAlg{X}\to\PreOrdAlg{X}$ as
  \[
    \psplit[X](u) := \saturate[X]\repr[X]{u} .
  \]
  For all $u,v∈\PreOrdAlg{X}$,
  \[
    u \obseqv[X] v
    \qquad\text{if and only if}\qquad
    \psplit[X](u) \obseqv[{\web{X}}] \psplit[X](v) .
  \]
  Hence $\psplit[X]$ is an injective map from $\OrdAlg{X}$ into
  $\OrdAlg{\web{X}}$.
  For all $u,v∈\OrdAlg{X}$,
  \[
    \psplit[X](u \psync v) = \psplit[X](u) \sync \psplit[X](v) .
  \]
\end{proposition}
\begin{proof}
  For compatibility with observational equivalences,
  first suppose that $u$ and $v$ are such that
  $\saturate\repr{u}\obseqv_{\web{X}}\saturate\repr{v}$.
  Consider a play $r∈\Plays{X}$, then we have
  $\outcome{u\psync r}=\outcome{\saturate\repr{u}\sync\repr{r}}
  =\outcome{\saturate\repr{v}\sync\repr{r}}=\outcome{v\psync r}$ 
  using Lemma~\ref{lemma:saturate}, so we have $u\obseqv[X]v$.

  Reciprocally suppose $u\obseqv[X]v$, then by definition for all play
  $r∈\Plays{X}$ we have $\outcome{u\psync r}=\outcome{v\psync r}$.
  By the remarks above, the outcome $\outcome{u\psync r}$ is equal to
  $\outcome{r\psync u}=\outcome{\repr{r}\sync\saturate\repr{u}}$,
  so we have
  $\outcome{\repr{r}\sync\saturate\repr{u}}=\outcome{\repr{r}\sync\saturate\repr{v}}$
  for all $r$.
  Let $s$ be an arbitrary play in $\Plays{X}$.
  Writing $u$ as a linear combination $\sum_{i∈I}λ_ir_i$, we get
  $
    \outcome{\saturate\repr{u}\sync s} =
    \sum_{i∈I} λ_i \outcome{\saturate\repr{r_i}\sync\repr{s}} .
  $
  If $\web{s}$ is not a representant subset of $\web{X}$, then this sum is
  zero since $\saturate{r_i}$ is a combination of plays whose supports are
  representant subsets.
  The same applies to $v$ so we have
  $\outcome{\saturate\repr{u}\sync s}=\outcome{\saturate\repr{v}\sync s}=0$.
  Now suppose that $\web{s}$ is a representant subset of $\web{X}$, then the
  representant $\repr{s}$ of $s$ has the same support as $s$ by definition, so
  there is a permutation $σ$ such that $σ s=\repr{s}$ and
  $σ\web{s}=\web{s}$.
  For all $i∈I$, if $\web{\repr{r_i}}=\web{s}$, then by definition of
  saturation we have $σ\saturate\repr{r_i}=\saturate\repr{r_i}$, so we
  get $\outcome{\saturate\repr{r_i}\sync s}
  =\outcome{\saturate\repr{r_i}\sync\repr{s}}$.
  If $\web{\repr{r_i}}\neq\web{s}$, then the equality holds trivially since
  both sides are $0$.
  By linearity, we can deduce $\outcome{\saturate\repr{u}\sync s}
  =\outcome{\saturate\repr{u}\sync\repr{s}}$, and applying the same reasoning
  to $v$, from our initial remarks we deduce
  $\outcome{\saturate\repr{u}\sync s}=\outcome{\saturate\repr{v}\sync s}$.
  Hence we get $\saturate\repr{u}\obseqv_{\web{X}}\saturate\repr{v}$.

  For the commutation property with synchronisation, consider two plays
  $r,s∈\Plays{X}$.
  If the supports $\web{\repr{r}}$ and $\web{\repr{s}}$ are distinct, then
  clearly $\psplit[X](r\psync s)=\psplit[X](r)\sync\psplit[X](s)=0$.
  Otherwise, let $A$ be this support.
  If $ρ$ is a permutation in $\pgroup[X]$ such that $ρ r=\repr{r}$, we
  have $\repr{r\psync s}=\repr{ρ(r\psync s)}=\repr{ρ r\psync s}
  =\repr{\repr{r}\psync s}$.
  Moreover, $\repr{r}\psync s=\repr{r}\sync\saturate\repr{s}$ so
  all terms in $\repr{r}\psync s$ have support $A$, and since for all play $t$
  with $\web{t}=A$ we have $\saturate\repr{t}=\saturate t$, we get
  \[
    \saturate \repr{r \psync s}
    = \saturate (\repr{r} \psync s)
    = \sum_{σ∈\rpgroup{A}} σ(\repr{r} \psync s)
    = \sum_{σ∈\rpgroup{A}} σ\repr{r} \psync s
    = \sum_{σ∈\rpgroup{A}}
      \sum_{\tau∈\rpgroup{A}} σ\repr{r} \sync \tau\repr{s}
    = \saturate\repr{r} \sync \saturate\repr{s}
  \]
  which concludes the proof.
\end{proof}

The commutation property could actually be written
$\psplit[X](u \psync v) = \psplit[X](u) \psync \psplit[X](v)$
since permuted and static synchronisations coincide in the static order
algebra $\OrdAlg{\web{X}}$, but we keep the notations distinct to stress the
fact that the second is static.
This establishes an injective morphism of $\Scal$-semialgebras, however this
morphism does not preserve outcomes: for a play $s$, we have
$\outcome{\psplit s}=\sharp(\rpgroup{\web{s}})\outcome{s}$; since this factor
depends on $\web{s}$, the outcome of $\psplit(u)$ is not even proportional to
that of $u$ in general.

\begin{proposition}\label{prop:units}
  Let $X$ be an arena.
  Assume $\Scal$ is rational.
  Then the $\Scal$-semialgebra $\OrdAlg{X}$ has a unit element if and only if
  the web $\web{X}$ is finite.
\end{proposition}
\begin{proof}
  Suppose $\web{X}$ is finite, then the set of plays $\Plays{X}$ is finite, so
  we can apply Lemma~\ref{lemma:partial-neutral} to the whole set $\Plays{X}$,
  which provides a vector $e$ and a non-zero integer $n$ such that
  $e\psync s=ns$ for all $s∈\Plays{X}$; then $e/n$ is a neutral element for
  synchronisation.
  Now suppose that $\web{X}$ is infinite.
  Let $u$ be an arbitrary vector in $\PreOrdAlg{X}$.
  Since $u$ is a finite linear combination of plays with finite support, there
  is an integer $n$ such that all non-zero components of $u$ are plays with
  supports of cardinal strictly less than $n$.
  Let $A$ be a subset of $\web{X}$ of cardinal $n$, then we must have
  $u\psync e_A=0\neq e_A$.
  This implies that no finite linear combination of plays can be neutral.
\end{proof}

\subsection{Bases} 
\label{sec:bases}

In this section, we describe the $\Scal$-semimodule $\OrdAlg{X}$ by providing
a subset of plays whose equivalence classes forms a basis.
Linear independence does not have a unique definition for modules over
arbitrary semirings~\cite{agg09:independence}, so we state the appropriate
definition for our needs, which clearly extends the standard one for vector
spaces:
\begin{definition}
  Let $\Scal$ be a semiring and $E$ a semimodule over $\Scal$.
  A family $(u_i)_{i∈I}$ in $E$ is linearly independent if, for any two
  families $(λ_i)_{i∈I}$ and $(\mu_i)_{i∈I}$ in $\Scal$ with
  finite support, if
  $\sum_{i∈I}λ_iu_i=\sum_{i∈I}\mu_iu_i$ then for all $i$,
  $λ_i=\mu_i$.
  A basis of $E$ is a linearly independent generating family.
\end{definition}

We first concentrate on the case of static order algebras.
The first thing we can remark about observational equivalence is that plays
of different supports are always independent, since compatibility
explicitly requires having the same support, so we have the following
decomposition:
\begin{proposition}\label{prop:direct-sum}
  For a finite static arena $X$, let $\PreOrdAlgS{X}$ be the submodule of
  $\PreOrdAlg{X}$ generated by plays of support $X$.
  Define the \emph{strict order algebra} over $X$ as the submodule
  $\OrdAlgS{X}$ of $\OrdAlg{X}$ made of equivalence classes of
  elements of $\PreOrdAlgS{X}$.
  Then for all static arena $X$ we have
  \[
    \OrdAlg{X} = \bigoplus_{Y∈\PartsFin{X}} \OrdAlgS{Y}.
  \]
\end{proposition}
\begin{proof}
  Clearly $\PreOrdAlg{X}$ is the direct sum of the $\PreOrdAlgS{Y}$, since
  this decomposition amounts to partitioning the basis $\Plays{X}$
  according to the supports $Y$ of its elements.
  As a consequence, $\OrdAlg{X}$ is the sum of the $\OrdAlgS{Y}$, and we have
  to prove that this sum is direct.
  Consider two vectors $u=\sum_{Y∈\PartsFin{X}}u_Y$ and
  $v=\sum_{Y∈\PartsFin{X}}v_Y$ such that $u\obseqv v$ and
  for all $Y∈\PartsFin{X}$, $u_Y,v_Y∈\PreOrdAlgS{Y}$ (necessarily, only
  finitely many of the $u_Y$ and $v_Y$ are not $0$).
  For each $Y∈\PartsFin{X}$, we have $u\sync e_Y=u_Y$ and $v\sync
  e_Y=v_Y$, so $u_Y\obseqv v_Y$ since $\sync$ is compatible with $\obseqv$.
  As a consequence, the decomposition of a vector in $\OrdAlg{X}$ on the
  submodules $\OrdAlgS{Y}$ is unique.
\end{proof}

We can thus focus on the study of strict order algebras.
These have the definite advantage of being finitely generated, since there are
finitely many different binary relations over a given finite set.
We will now provide explicit bases for them, depending on the
structure of $\Scal$.

Let $X$ be a finite static arena.
Clearly, for all inconsistent plays $r$ we have $r\obseqv0$, so we can
consider only consistent plays, \ie\ plays $r$ such that $\order{r}$ is an
order relation.
In the following statements, as a slight abuse of notations, a play $r$ with
$\web{r}=X$ is identified with its order relation $\order{r}$, and also
with its equivalence class in $\OrdAlgS{X}$.
Let $\Ord{X}$ be the set of all partial order relations over $X$.

The notations $<_r$, $≥_r$, $>_r$ are defined as expected.
We denote by $\conc_r$ the incomparability relation: $x\conc_ry$ if and
only if neither $x≤_ry$ nor $y≤_rx$.
We write $x\conceq_ry$ if $x=y$ or $x\conc_ry$.
If there is no ambiguity, we may omit the subscript $r$ in these notations.
The notation $[a<b]_X$, for $a,b∈X$, represents the smallest partial order
over $X$ for which $a<b$, that is $\identity{X}\cup\{(a,b)\}$.
The notation extends to more complicated formulas, for instance $[a<b,c<d]$ is
the smallest partial order for which $a<b$ and $c<d$.
We write $r\comp s$ to denote that two partial orders $r$ and $s$ are
compatible.

\begin{proposition}\label{prop:totals}
  Let $\Totals{X}$ be the set of total orders over $X$, then
  $\Totals{X}$ is a linearly independent family in $\OrdAlgS{X}$.
\end{proposition}
\begin{proof}
  We prove the equivalent statement that two observationally equivalent
  combinations of total orders are necessarily equal.
  Let $u=\sum_{t∈\Totals{X}}λ_tt$ and $v=\sum_{t∈\Totals{X}}\mu_tt$
  be two combinations such that $u\obseqv v$.
  If $r$ and $s$ are two distinct total orders over $X$, there exists a pair
  $(a,b)∈X^2$ such that $a<_rb$ and $b<_sa$, hence $r$ and $s$ are not
  compatible, so $\inter{r}{s}=0$.
  Besides, it always holds that $\inter{r}{r}=1$,
  so for all $t∈\Totals{X}$, $\inter{u}{t}=λ_t$ and
  $\inter{v}{t}=\mu_t$, so $u\obseqv v$ implies $λ_t=\mu_t$ for all $t$,
  hence $u=v$.
\end{proof}

However, in general, $\Totals{X}$ is not a generating family for $\OrdAlgS{X}$.
The simplest counter-example can be found if $X$ has two points.
Write $X=\{a,b\}$, then $\Ord{X}$ has three elements:
\[
  \Ord{\{a,b\}}=\bigl\{[a\conc b],[a<b],[a>b]\bigr\}.
\]
Then in the canonical basis $([a\conc b],[a<b],[a>b])$ of $\PreOrdAlgS{X}$,
the matrix of $(u,v)\mapsto\inter{u}{v}$ is
\[
  \begin{pmatrix}
    1 & 1 & 1 \\
    1 & 1 & 0 \\
    1 & 0 & 1
  \end{pmatrix}
\]
If $\Scal$ is the field of reals, for instance, then this matrix is
invertible, which means that the three orders are linearly independent, hence
$\OrdAlgS{X}$ is isomorphic to $\Scal^{\Ord{X}}$ (this isomorphism holds if
and only if the cardinal of $X$ is at most $2$, as we shall see below).
There is one case where $[a<b]$ and $[b<a]$ do generate $\OrdAlgS{\{a,b\}}$,
namely when addition in $\Scal$ is idempotent, \ie\ when $1+1=1$.

\begin{proposition}\label{prop:basis-totals}
  $\Totals{X}$ is a basis of $\OrdAlgS{X}$ for all $X$
  if and only if addition in $\Scal$ is idempotent.
\end{proposition}
\begin{proof}
  By Proposition~\ref{prop:totals}, we know that $\Totals{X}$ is always a
  linearly independent family, so all we have to prove is that it generates
  $\OrdAlgS{X}$ if and only if $1+1=1$ in $\Scal$.

  Firstly, assume that $\Totals{X}$ generates $\OrdAlgS{X}$ for all $X$.
  Then, for $X=\{a,b\}$, there are two scalars $λ,\mu∈\Scal$ such that
  $[a\conc b]\obseqvλ[a<b]+\mu[a>b]$.
  Then we have
  \[
    \inter{[a\conc b]}{[a<b]}
    = λ\inter{[a<b]}{[a<b]} + \mu\inter{[a>b]}{[a<b]} = λ
  \]
  but by definition we have $\inter{[a\conc b]}{[a<b]}=1$, so $λ=1$.
  Similarly, we get $\mu=1$, and so $[a\conc b]=[a<b]+[a>b]$.
  As a consequence, we have
  \[
    1 = \inter{[a\conc b]}{[a\conc b]}
    = \inter{[a<b]}{[a\conc b]} + \inter{[a>b]}{[a\conc b]}
    = 1 + 1 .
  \]

  Reciprocally, assume $\Scal$ satisfies $1+1=1$.
  Let $X$ be an arbitrary finite set and let $r∈\Ord{X}$.
  Let $u=\sum_{i=1}^kt_i$ be the sum of all total orders that are compatible
  with $r$.
  Consider an arbitrary order $s∈\Ord{X}$.
  Then we have $\inter{u}{s}=\sum_{i=1}^k\inter{t_i}{s}$ and each term of this
  sum is $0$ or $1$.
  If $s$ is compatible with $r$, then there is a total order $t$ that extends
  both $r$ and $s$, so $t$ is one of the $t_i$;
  since $s$ and $t$ are compatible, the sum contains at least one $1$ so
  $\inter{u}{s}=1=\inter{r}{s}$.
  If $s$ is incompatible with $r$, then it is incompatible with any order that
  contains $r$, and in particular it is incompatible with all the $t_i$, so
  $\inter{u}{s}=0=\inter{r}{s}$.
  As a consequence we have $r\obseqv u$, which proves that $\Totals{X}$
  generates $\OrdAlgS{X}$.
\end{proof}

%

In the general case, without any hypothesis on the semiring $\Scal$, it
happens that the family of all orders over $X$ is not linearly independent, as
soon as $X$ has at least three points.

\begin{proposition}\label{prop:split}
  For all semiring $\Scal$,
  in $\PreOrdAlgS[\Scal]{\{x,y,z\}}$ we have
  \[
    \tikzformula{
      \eventnode{(0,0)}{x}{left}{x}
      \eventnode{(0,2)}{y}{left}{y}
      \eventnode{(1,1)}{z}{right}{z}
      \draw (x) -- (y); }
    + \tikzformula{
      \eventnode{(0,0)}{x}{left}{x}
      \eventnode{(0,2)}{y}{left}{y}
      \eventnode{(0,1)}{z}{right}{z}
      \draw (x) -- (y); }
    = \tikzformula{
      \eventnode{(0,0)}{x}{left}{x}
      \eventnode{(0,2)}{y}{left}{y}
      \eventnode{(1,1)}{z}{right}{z}
      \draw (z) -- (x) -- (y); }
    + \tikzformula{
      \eventnode{(0,0)}{x}{left}{x}
      \eventnode{(0,2)}{y}{left}{y}
      \eventnode{(1,1)}{z}{right}{z}
      \draw (x) -- (y) -- (z); }
  \]
\end{proposition}
\begin{proof}
  We use the following notations:
  $a:=[x<y]$,
  $b:=[x<z<y]$,
  $c:=[x<y,x<z]$,
  $d:=[x<y,z<y]$,
  so that the equation we prove is $a+b=c+d$.
  Let $s$ be a partial order over $\{x,y,z\}$.
  First remark that $s\comp a$ if and only if $s\comp c$ or
  $s\comp d$.
  Indeed, assume $s\comp a$, then there is a total order $t$ that contains $a$
  and $s$.
  If $x<_tz$ then $[x<z]\subseteq t$ so $a\sync{}[x<z]\subseteq t$, hence
  $s\comp a\sync{}[x<z]=c$.
  Otherwise $z<_tx<_ty$ so $z<_ty$ then $s\comp d$.
  Reciprocally, if $s\comp c$ or $s\comp d$ then $s\comp a$ since
  $a$ is included in $c$ and $d$.
  Secondly, remark that $s\comp b$ if and only if $s\comp c$ and
  $s\comp d$.
  Indeed, assume that $s\comp c$ and
  $s\comp d$.
  Let $s'=s\sync c=s\sync a\sync{}[x<z]$.
  Suppose $s'\incomp [z<y]$, then $y<_{s'}z$.
  By hypothesis we cannot have $y<_{a\sync s}z$, so $(y,z)$ occurs in $s'$ but
  not in $(s\sync a)\cup[x<z]$, which implies $y<_{s\sync a}x$.
  This contradicts the hypothesis $x<_ay$, hence $s'\comp[z<y]$, so
  $s\comp a\sync{}[x<z]\sync{}[z<y]=b$.
  The reciprocal implication is immediate since $a\subseteq b$ and
  $d\subseteq b$.
  As a consequence of the two remarks above, we have $\inter{a}{s}=1$ if and
  only if $\inter{c}{s}=1$ or $\inter{d}{s}=1$, which is equivalent
  to $\inter{(c+d)}{s}∈\{1,2\}$.
  Moreover, $\inter{(c+d)}{s}=2$ if and only if $\inter{c}{s}=1$
  and $\inter{d}{s}=1$, which is equivalent to $\inter{b}{s}=1$.
  Therefore $\inter{(a+b)}{s}=\inter{(c+d)}{s}$.
\end{proof}

This proposition applies to orders on three points, but the exact same
argument applies in any larger context, since the proof never uses the fact
that there are no other points than $x,y,z$.
So for any play $r$ and points $x,y,z∈\web{r}$ such that $x<_ry$,
$x\conc_rz$ and $y\conc_rz$ we have
\[
  r + (r\sync{}[x<z<y]) \obseqv (r\sync{}[x<z]) + (r\sync{}[z<y]) .
\]
This can also be deduced from Proposition~\ref{prop:split} using the partial
composition operators defined in Section~\ref{sec:partial}.
When $\Scal$ is a ring, it allows us to express each of
the patterns of the equation in Proposition~\ref{prop:split} as a linear
combination of the others with coefficients $1$ and $-1$.
This implies that for each of these patterns, the set of all orders over $X$
that do not contain the considered pattern generates $\OrdAlgS{X}$.
In each case, the forbidden pattern defines a particular class of orders,
respectively weak total orders (as of Proposition~\ref{prop:weak} below),
orders of height at most $2$ and forests with roots up or down.

\begin{proposition}\label{prop:weak}
  Let $(X,≤)$ be a partially ordered set.
  The following conditions are equivalent:
  \begin{itemize}
  \item 
    For all $x,y,z∈X$, if $x<y$ then $x<z$ or $z<y$.
  \item 
    The relation $\conceq$ is an equivalence.
  \item 
    There is a totally ordered set $(Y,≤)$ and a function $f:X\to Y$ such
    that, for all $x,y∈X$, $x<y$ if and only if $f(x)<f(y)$.
  \end{itemize}
  Let $\Weak{X}$ be the set of orders that satisfy these conditions,
  called \emph{weak total orders} over~$X$.
\end{proposition}
\begin{proof}
  Firstly, assume that for all $x,y,z∈X$, if $x<y$ then
  $x<z$ or $z<y$.
  It is clear that $\conceq$ is always reflexive and symmetric.
  Let $x,y,z∈X$ such that $x\conceq z$ and $z\conceq y$.
  If $x<y$, then by hypothesis we must have $x<z$ or $z<y$, which
  contradicts the hypothesis on $x,y,z$.
  Similarly we cannot have $y<x$, so $x\conceq y$.
  Therefore $\conceq$ is transitive and it is an equivalence relation.

  Secondly, assume $\conceq$ is an equivalence relation.
  Let $Y$ be the set of equivalence classes of~$\conceq$.
  Define the relation $\sqsubseteq$ on $Y$ as $A\sqsubseteq B$ if $a≤ b$
  for some $a∈A$ and $b∈B$.
  The relation $\sqsubseteq$ is reflexive since for all $A∈Y$, for any
  $a∈A$ we have $a≤ a$ so $A\sqsubseteq A$.
  Assume $A\sqsubseteq B$ and $B\sqsubseteq A$ for some $A,B∈Y$, then there
  are $a,a'∈A$ and $b,b'∈B$ such that $a≤ b$ and $b'≤ a'$; if
  $a<b'$ then $a<a'$ which contradicts $a\conceq a'$, similarly if $b'<a$ then
  $b'<b$ which contradicts $b'\conceq b$, so $a\conceq b'$, which implies that
  $A$ and $B$ are the same class, therefore $\sqsubseteq$ is antisymmetric.
  Assume $A\sqsubseteq B$ and $B\sqsubseteq C$ for some $A,B,C∈Y$, then
  there are $a∈A$, $b,b'∈B$ and $c∈C$ such that $a≤ b$ and
  $b'≤ c$; if $a\conceq c$ then $A=C$ hence $A\sqsubseteq C$, otherwise we
  must have $a<c$ or $c<a$, but the second case implies $b'≤ c<a≤ b$
  which contradicts $b\conceq b'$, so $a<c$ and $A\sqsubseteq C$, hence
  $\sqsubseteq$ is transitive.
  Totality is immediate: if $A$ and $B$ are two distinct classes, then every
  pair $(a,b)∈A× B$ is comparable.
  Let $f$ be the function that maps each element of $X$ to its class.
  If $x<y$ then $f(x)\sqsubset f(y)$ by definition.
  Reciprocally, if $f(x)\sqsubset f(y)$, then $x$ and $y$ must be comparable
  (since they are in distinct classes), and $y<x$ would imply
  $f(y)\sqsubset f(x)$, so $x<y$.

  Finally, assume there is $f:X\to Y$ where $Y$ is totally ordered such that
  $x<y$ if and only if $f(x)<f(y)$.
  Let $x,y,z$ be such that $x<y$, then $f(x)<f(y)$.
  Since the order on $Y$ is total, we must have either $f(x)<f(z)$ or
  $f(z)<f(y)$ (or both), hence $x<z$ or $z<y$.
\end{proof}

In other words, a weak total order is a total order over sets of mutually
incomparable points.
Interestingly, this kind of order was considered long ago in scheduling
theory~\cite{lam78:time} as the possibility to label events with time stamps in
a possibly non-injective manner.
It turns out that weak total orders form a basis.

\begin{definition}
  Let $r∈\Ord{X}$.
  Two elements $a,b∈X$ are equivalent in $r$, written $a\sim_rb$,
  if for all $c∈X∖\{a,b\}$, $a<_rc$ if and only if $b<_rc$, and
  $c<_ra$ if and only if $c<_rb$.
  For a pair $a\sim_rb$ with $a\neq b$, let $r/(a\sim b)$ be the order
  $r\cap{(X∖\{b\})^2}$ over $X∖\{b\}$.
\end{definition}

\begin{definition}
  Let $a,b∈X$ with $a≠b$.
  For each $r∈\Ord{X∖\{b\}}$, define the relations
  \begin{formulas}
    r_{a\sim b} := r ∪ \set{(x,b)}{(x,a)∈r} ∪ \set{(b,x)}{(a,x)∈r}, \break
    r_{a<b} := r_{a\sim b} ∪ \{(a,b)\}, \split
    r_{a>b} := r_{a\sim b} ∪ \{(b,a)\}.
  \end{formulas}
\end{definition}

Clearly, $r_{a\sim b}$, $r_{a<b}$ and $r_{a>b}$ are partial orders over
$X$ in which $a$ and $b$ are equivalent.

\begin{lemma}\label{lemma:coh-up}
  Let $a,b$ be two distinct elements of $X$.
  For all $r∈\Weak{X}$ and $s∈\Ord{X∖\{b\}}$,
  \begin{itemize}
  \item 
    if $a<_rb$ then $r\comp s_{a\sim b}$ if and only if $r\comp s_{a<b}$,
    moreover $r\incomp s_{a>b}$,
  \item 
    if $a>_rb$ then $r\comp s_{a\sim b}$ if and only if $r\comp s_{a>b}$,
    moreover $r\incomp s_{a<b}$,
  \item 
    if $a\conc_rb$ then $r\comp s_{a\sim b}$
    if and only if $r\comp s_{a<b}$
    if and only if $r\comp s_{a>b}$.
  \end{itemize}
\end{lemma}
\begin{proof}
  If $a<_rb$, we have $r\cup s_{a\sim b}=r\cup s_{a<b}$, since
  $s_{a\sim b}$ and $s_{a<b}$ only differ on $(a,b)$, so the compatibility
  of the two pairs is equivalent to this union being acyclic.
  The same argument applies to the case $a>_rb$.
  For the case $a\conc_rb$, first assume $r\comp s_{a\sim b}$ and let
  $t=r\sync s_{a\sim b}$.
  By definition of weak orders, we have $a\sim_rb$.
  If $a<_tb$ then there exists a sequence $a=a_0,\ldots,a_n=b$ such that for
  each $i<n$, $a_i<_ra_{i+1}$ or $a_i<_{s_{a\sim b}}a_{i+1}$,
  but since $a$ and $b$ are equivalent in both $r$ and $s_{a\sim b}$, we
  can replace $b$ with $a$ in this sequence, which leads to the contradiction
  $a<_ta$.
  By the same argument we cannot have $b<_ta$, so $a\conc_tb$.
  We thus have $t\comp[a<b]$ hence $s_{a<b}=s_{a\sim b}\sync{}[a<b]\comp r$,
  and similarly $r\comp s_{a>b}$.
  The reverse implications are immediate since $s_{a\sim b}$ is included in
  both $s_{a<b}$ and $s_{a>b}$.
\end{proof}

\begin{proposition}\label{prop:basis-weak}
  If $\Scal$ is a ring, then for all finite set $X$, $\Weak{X}$ is
  a basis of $\OrdAlgS{X}$.
\end{proposition}
\begin{proof}
  Let $Z(X)$ be the submodule of all the $u∈\PreOrdAlgS{X}$ such that for
  all order $r$ over $X$, $\inter{u}{r}=0$.
  We actually prove the fact that $\PreOrdAlgS{X}$ is isomorphic to the direct
  sum $\Scal^{\Weak{X}}⊕ Z(X)$, which is equivalent since by definition
  $\OrdAlgS{X}$ is $\PreOrdAlgS{X}/Z(X)$ when the semiring $\Scal$ is a ring.

  We first prove that for all order $r$ over $X$ there is an
  $s∈\Scal^{\Weak{X}}$ such that $r-s∈Z(X)$.
  Let $N(r)=\set{(a,b,c)∈X^3}{a<_rb,a\conc_r c,b\conc_r c}$,
  we proceed by induction on $\sharp N(r)$.
  If $r=\emptyset$, then by Proposition~\ref{prop:weak} we
  have $r∈\Weak{X}$, so we can set $s=r$.
  Otherwise, consider a triple $(a,b,c)∈N(r)$.
  Define the orders $r_1:=r\sync{}[a<c]$, $r_2:=r\sync{}[c<b]$ and
  $r_3:=r\sync{}[a<c<b]$.
  By Proposition~\ref{prop:split}, we have $r_1+r_2-r_3-r∈Z(X)$.
  Besides, for each $i∈\{1,2,3\}$, clearly $N(r_i)\subset N(r)$ and
  $(a,b,c)∈N(r)∖ N(r_i)$, so $\sharp N(r_i)<\sharp N(r)$.
  We can then apply the induction hypothesis to get an $s_i∈\Scal^{\Weak{X}}$
  such that $r_i-s_i∈Z(X)$.
  We can then conclude by setting $s:=s_1+s_2-s_3$.

  As a consequence we have $\PreOrdAlgS{X}=\Scal^{\Weak{X}}+Z(X)$, and we now
  prove that this sum is direct by proving $\Scal^{\Weak{X}}\cap Z(X)=\{0\}$.
  We proceed by recurrence on the size of $X$.
  If $X$ has $0$ or $1$ element, then the only order over $X$ is the trivial
  order $t$, and $\inter{t}{t}=1\neq0$, so $Z(X)=\{0\}$ and the result
  trivially holds.
  Now let $n≥ 2$, suppose the result holds for all $X$ with at most $n-1$
  points, and let $u∈\Scal^{\Weak{X}}\cap Z(X)$.
  We now prove that $u$ is the zero function.

  Let $r$ be a weak total order that is not a total order, let $a,b∈X$ such
  that $a\conc_rb$.
  Let $X'=X∖\{b\}$.
  Define $u'∈\Scal^{\Weak{X'}}$ by $u'(t)=u(t_{a\sim b})$ for all
  $t\in\Ord{X'}$, so that $u(r)=u'(r/(a\sim b))$.
  For any orders $s∈\Weak{X}$ and $t∈\Ord{X'}$, by
  Lemma~\ref{lemma:coh-up} we have that
  $\inter{s}{(t_{a<b}+t_{a>b}-t_{a\sim b})}$ is $0$ if $a$ and $b$ are comparable in $s$, otherwise it is equal to
  $\inter{s}{t_{a\sim b}}$, which is itself equal to $\inter{s/(a\sim b)}{t}$
  by restriction to $X'$.
  Let $s'=s/(a\sim b)$, we have
  \[
    \inter{u}{(t_{a<b}+t_{a>b}-t_{a\sim b})}
    = \sum_{s∈\Weak{X},a\conc_sb}u(s)\inter{s}{t_{a\sim b}}
    = \sum_{s∈\Weak{X},a\conc_sb}u'(s')\inter{s'}{t}
  \]
  The mapping $s\mapsto s/(a\sim b)$ is a bijection from weak total orders
  over $X$ such that $a\conc b$ to weak total orders over $X'$, so the
  latter sum is equal to
  $\sum_{s'∈\Weak{X'}}u'(s')\inter{s'}{t}=\inter{u'}{t}$.
  Besides, $u$ is in $Z(X)$ so $\inter{u}{(t_{a<b}+t_{a>b}-t_{a\sim b})}=0$,
  which implies $\inter{u'}{t}=0$.
  This holds for all $t$, so $u'∈Z(X')$.
  By construction we have $u'∈\Scal^{\Weak{X'}}$ so $u'$ is in
  $\Scal^{\Weak{X'}}\cap Z(X')$.
  By the induction hypothesis this is $\{0\}$, so $u'=0$ and as a consequence
  we have $u(r)=u'(r/(a\sim b))=0$.

  By the argument above, we thus know that $u(r)=0$ as soon as $r$ is not a
  total order.
  In other words, $u$ is a linear combination of total orders.
  From Proposition~\ref{prop:totals} we know that total orders are linearly
  independent in $\OrdAlgS{X}$, so we can conclude that $u=0$.
\end{proof}

As a consequence, weak total orders on subsets of $\web{X}$ form a basis of
the static order algebra $\OrdAlg{X}$.
We can extend this property to arbitrary order algebras using the
representation property.

\begin{theorem}\label{thm:basis}
  Let $X$ be an arena.
  Then $\OrdAlg{X}$ has a basis $(b_i)_{i∈I}$ made of plays if
  \begin{itemize}
  \item $\Scal$ is idempotent, then the $b_i$ are the orbits of
    totally ordered plays under $\pgroup[X]$, or
  \item $\Scal$ is a regular ring, then the $b_i$ are the orbits of
    weakly totally ordered plays under $\pgroup[X]$.
  \end{itemize}
  In both cases, if $\Scal$ is rational, then there exists a family of vectors
  $(b^*_i)_{i∈I}$ such that for all $i,j∈I$, $\outcome{b_i\psync b^*_j}$ is
  $1$ if $i=j$ and $0$ otherwise.
\end{theorem}
\begin{proof}
  Propositions \ref{prop:basis-totals} and \ref{prop:basis-weak} provide bases
  of the appropriate kinds for strict static order algebras.
  By Proposition~\ref{prop:direct-sum}, these yield bases for static order
  algebras.
  In each case, call the elements of these bases \emph{base plays}.
  A permutation of a (weak) total order is always an order of the same kind,
  so from the fact that base plays generate $\OrdAlg{\web{X}}$, we deduce that
  they also generate $\OrdAlg{X}$.
  Now consider two linear combinations $u=\sum_{i∈I}λ_ir_i$ and
  $v=\sum_{i∈I}\mu_ir_i$, where the $r_i$ are distinct base plays for
  $\web{X}$ and representants (as of Definition~\ref{def:choice}), and suppose
  $u\obseqv v$.
  By Proposition~\ref{prop:represent}, we can deduce
  $\sum_{i∈I}λ_i\saturate r_i
  \obseqv[\web{X}]\sum_{i∈I}\mu_i\saturate r_i$,
  and this equivalence is an equality since both sides are linear combinations
  of base plays.
  Now consider any $i∈I$.
  In $\sum_{i∈I}λ_i\saturate r_i$, the coefficient of $r_i$ is
  $\multiplicity[X]{r_i}$, so the equality above implies
  $\multiplicity[X]{r_i}λ_i=\multiplicity[X]{r_i}\mu_i$, and
  subsequently $λ_i=\mu_i$ since $\multiplicity[X]{r_i}$ is a non-zero integer
  and $\Scal$ is regular.
  Hence representants of base plays form a basis of $\OrdAlg{X}$.

  If $\Scal$ is an idempotent semiring, then by Proposition~\ref{prop:basis-totals}
  the family $(b_i)_{i∈I}$ is made of total orders, so if we set
  $b^*_i=b_i$ for each $i$ we have the expected property.

  Now suppose $\Scal$ is a rational ring.
  Let $A$ be a representant finite subset of $\web{X}$.
  Call $a_1,\ldots,a_n$ the subset of the basis whose plays have support $A$,
  and let $M=(m_{ij})$ be the $n× n$ matrix such that
  $m_{ij}=\outcome{a_i\psync a_j}$.
  $M$ has coefficients in natural numbers, and since the family $(a_i)$ is
  linearly independent by hypothesis, $M$ is invertible in $\Rat$.
  Since $\Scal$ is a regular ring, it is an algebra over $\Rat$, so $M$ is
  also invertible in $\Scal$.
  Let $M^{-1}=(m'_{ij})$ and let $a_i^*:=\sum_{j=1}^nm'_{ij}a_i$, then by
  construction $\outcome{a_i\psync a^*_j}$ is $1$ if $i=j$ and $0$ otherwise.
\end{proof}

Observe that if $\Scal$ is a rational ring, then in particular it is an algebra
over $\Rat$, then $\OrdAlg[\Scal]{X}=\Scal\otimes\OrdAlg[\Rat]{X}$ as
$\Rat$-algebras, since all plays decompose uniquely as linear combinations of
base plays with integer coefficients.
The outcome in $\OrdAlg[\Scal]{X}$ then appears as the tensor of the identity
over $\Scal$ and the outcome over $\OrdAlg[\Rat]{X}$.
The algebra $\OrdAlg[\Rat]{X}$ further decomposes into the direct sum of the
strict order algebras $\OrdAlgS[\Rat]{Y}$ for all representant subset $Y$ with
the permutation group induced over it.
This is particularly useful since the $\OrdAlgS[\Rat]{Y}$ are finite
dimensional vector spaces over $\Rat$.

On the other hand, if $\Scal$ is neither idempotent nor a regular ring, it is
possible that there is no base.
For instance, if $\Scal=\Nat$, then clearly a play $s$ cannot be decomposed as
a non-trivial sum of vectors, so any generating family must contain all plays,
but then the equation of proposition~\ref{prop:split} states that they are not
linearly independent.

\section{Logical structure} 
\label{sec:logic}

In this section, we describe constructions on order algebras.
Although order algebras themselves have some interesting structure, the
actual objects we are interested in are submodules of such algebras, hereafter
called types, which enjoy better properties.

\begin{definition}
  A \emph{type} over an arena $X$ is a submodule of $\OrdAlg{X}$ generated by
  a family of plays in $\Plays{X}$.
  A type is \emph{strict} if it does not contain the empty play.
  The notation $A:X$ is used to represent the fact that $A$ is a type over $X$.
  A \emph{morphism} between types $A:X$ and $B:Y$ is a linear map $f$ from
  $\OrdAlg{X}$ to $\OrdAlg{Y}$ such that $f(A)\subset B$.
\end{definition}

The requirement that types are generated by plays is justified by the idea
that a type should be a constraint on the behaviours of processes, and that
such a constraint should boil down to a constraint on the shape of plays that
a process can exhibit.
We could also define a type over $X$ simply as a subset $S$ of $\Plays{X}$,
but the definition as submodules makes it clear that observationally
equivalent vectors should belong to the same types, even if one is a
combination of plays in $S$ while the other is not (this can happen even if
$S$ is closed under permutations, because of the equation of
Proposition~\ref{prop:split}).

\begin{example}
  The intended meaning of order algebras is that vectors, that is linear
  combinations of plays, represent processes.
  Then types impose constraints on the possible behaviours of processes, based
  on the possible interactions scenarii they may exhibit.
  For instance, we can define the type of processes that perform three actions
  of label $a$, as the submodule generated by the plays that contain three
  points in the orbit $a$.
  Similarly, we could define the type of all plays that include as many $a$'s
  as $b$'s.

  Typed may also used in particular to impose well-formedness conditions.
  For instance, when modelling a calculus like \pii\ that includes
  communication of bound names, one wants to impose that any play that
  contains an event on a bound name also contains the event that communicates
  this name.
\end{example}

\begin{example}\label{ex:type}
  Note that the condition of being a submodule of $\OrdAlg{X}$ imposes
  non-trivial conditions.
  For instance, consider the CCS algebra, as of Example~\ref{ex:ccs}, with
  $\Scal$ idempotent.
  We can define the type of processes in which all actions $a$ are causally
  independent of all actions $b$, as generated by the plays where all
  occurrences of $a$ are incomparable with all occurrences of $b$.
  This type does not contain the processes $a.b$ and $b.a$, obviously, but it
  does contain their sum $a.b+b.a$, which is observationally equivalent to
  the parallel composition $a\para b$.
\end{example}

\begin{proposition}\label{prop:type-basis}
  If $\Scal$ is a rational ring, then for all type $A:X$ there is a family of
  plays $(c_i)_{i∈I}$ and a family of vectors $(c^*_i)_{i∈I}$ in
  $\OrdAlg{X}$ such that $(c_i)_{i∈I}$ is a basis of $A$ and for all
  $i,j∈I$, $\outcome{c_i\psync c^*_j}$ is $1$ if $i=j$ and $0$ otherwise.
\end{proposition}
\begin{proof}
  Recall that if $\Scal$ is a rational ring, then it is an algebra over $\Rat$
  and $\OrdAlg[\Scal]{X}$ can be seen as the $\Rat$-algebra
  $\Scal\otimes\OrdAlg[\Rat]{X}$.
  Since $A$ is generated by plays, we can then decompose it as
  $\Scal\otimes A'$ for a type $A'$ in $\OrdAlg[\Rat]{X}$, so it is enough to
  prove the result in the case $\Scal=\Rat$.
  In this case $A$ is a subspace of the vector space $\OrdAlg[\Rat]{X}$, so it
  is a standard result that from the generating family we can extract a basis.

  Now assume that $(c_i)_{i∈I}$ is a basis of $A$, and consider a
  particular base play $c_n$.
  Set $J:=\set{i∈I}{\web{c_i}=\web{c_n}}$.
  Then $(c_i)_{i∈J}$ is a basis of the intersection of $A$ and
  $\OrdAlgS{\web{c_n},\rpgroup{\web{c_n}}}$, the strict order algebra over
  $\web{c_n}$ with the induced permutation group, which is a
  finite-dimensional $\Rat$-vector space.
  Let $f$ be a linear form over this algebra such that $f(c_n)=1$ and for all
  $j∈J∖\{i\}$, $f(c_j)=0$.
  Using the bases $(b_n)$ and $(b^*_n)$ from Theorem~\ref{thm:basis} we can
  define $c^*_n=\sum_{\web{b_k}=\web{c_n}}f(b_k)b^*_k$ and check that for all
  vector $x$ with $\web{x}=\web{c_n}$ we have $f(x)=\outcome{c^*_n\psync x}$.
  Then $c^*_n$ satisfies the expected condition.
\end{proof}

\subsection{Products and linear maps} 
\label{sec:partial}

When combining order algebras, we need a notion of combination of arenas.
Disjoint union is the simplest way, and also the most sensible one:
\begin{definition}
  Let $(X_i)_{i\in I}$ be a family of arenas with pairwise disjoint webs.
  Define the sum of the family $(X_i)$ as
  \[
    \sum_{i∈I} X_i :=
    \Bigl( \biguplus_{i\in I} X_i, \prod_{i∈I} \pgroup[X_i] \Bigr)
  \quad \text{with} \quad
    σ⋅x := σ_j⋅x
  \quad \text{for all }
    σ∈\prod_{i∈I} \pgroup[X_i] \text{ and }
    x∈\web{X_j} .
  \]
  We use equivalently the infix notation $X_1+X_2+\cdots+X_n$ for finite sums.
\end{definition}

\begin{example}\label{ex:sum-csp}
  Following on our first examples, if $X_A$ and $X_B$ are the arenas used for
  modelling CSP processes over alphabets $A$ and $B$ respectively (see
  Example~\ref{ex:csp}), then assuming $A$ and $B$ are disjoint the webs
  $\web{X_A}$ and $\web{X_B}$ are disjoint too and $X_A+X_B$ is actually the
  arena for processes in the alphabet $A\uplus B$.
\end{example}

This sum can be seen as a coproduct in a suitable category of arenas.
At the level of order algebras, however, this operation is not a Cartesian
product or coproduct, and not even a tensor product in the sense of
$\Scal$-algebras, because the algebra $\OrdAlg{X+Y}$ contains more plays than
those that appear as disjoint unions of a play in $\web{X}$ and one in
$\web{Y}$.
However, $\OrdAlg{X+Y}$ contains products and tensors as submodules, hence our
definition of types.
In all statements below, unless explicitly stated, different arenas are always
supposed to be disjoint.

\begin{proposition}
  For all types $A:X$ and $B:Y$, $A+B$ is a type over $X+Y$ that is isomorphic
  to the direct sum and Cartesian product of $A$ and $B$.
\end{proposition}
\begin{proof}
  $A$ is a submodule of $\OrdAlg{X}$, which is itself obviously a submodule of
  $\OrdAlg{X+Y}$.
  Similarly, $B$ is a submodule of $\OrdAlg{X+Y}$, and since $X$ and $Y$ are
  disjoint, so are $A$ and $B$, since no permutation in $X+Y$ can map a point
  of $\web{X}$ to a point of $\web{Y}$.
  Hence the submodule generated by $A+B=\Vect{A\cup B}$ in $\OrdAlg{X+Y}$ is a
  direct sum of $A$ and $B$.
\end{proof}

\begin{definition}
  Let $X$ be an arena, let $Y$ be a subset of $\web{X}$ closed under
  permutations in $\pgroup[X]$.
  Restriction to $Y$ is the linear map $\restric{Y}{}$ over $\PreOrdAlg{X}$
  such that for all $r∈\Plays{X}$,
  \[
    \restric{Y}{r} :=
    \outcome{r} ⋅ (\web{r}\cap Y, \order{r}\cap Y^2)
  \]
\end{definition}

Restriction of a play $r$ to a given subset $Y\subset\web{X}$
amounts to ignore the part of $r$ that happens outside $Y$, considering that
events in $\web{X}∖ Y$ are private, hence unobservable.
The fact that $Y$ must be closed under permutations is in accordance with the
intuition that two plays are indistinguishable when they are permutations of
each other.

\begin{example}
  Following on example~\ref{ex:sum-csp},
  $\restric{X_A}:\PreOrdAlg{X_A+X_B}→\PreOrdAlg{X_A}$
  precisely represents CSP's restriction operator that maps a trace $t$ to the
  restricted trace $t\upharpoonright A$.
\end{example}

Hence, as we shall see (in detail in Section~\ref{sec-traces}), restriction
does \emph{not} correspond to the hiding operator $\new{}$ of the
\pii-calculus and related languages.
Indeed, the externally observable behaviours of $\new{u}P$ are those of $P$
that do not involve an event on $u$, which to mapping to $0$ all plays in $P$
that contain an event on $u$, before actually restricting to the arena that
does not contain events on $u$.

\begin{remark}
  Note that in the definition we impose a coefficient $\outcome{r}$ on the
  restricted play.
  Since the outcome $\outcome{r}$ is $0$ or $1$ for any play $r$, this amounts
  to imposing that $\restric{Y}{r}$ be $0$ if $r$ is inconsistent.
  This condition is necessary because we want outcomes to be preserved by
  restriction: if a play is inconsistent, it means that it contains some
  deadlock, and hiding the place where this occurs surely should not resolve
  the deadlock.
  For instance, an inconsistent play like
  $(\tikzinline{
    \eventnode{(0,0)}{a}{left}{a}
    \eventnode{(1,0)}{b}{right}{b}
    \draw (a) to [bend left] (b) to [bend left] (a); })$,
  when restricted to $\{a\}$, would yield the consistent play
  $(\tikzinline{ \eventnode{(0,0)}{a}{left}{a} })$.

  Incidentally, this implies that $\restric{X}{}:\PreOrdAlg{X}→\PreOrdAlg{X}$
  is not the identity, because it collapses all inconsistent plays to $0$.
  However, up to observational equivalence, it is the identity.
\end{remark}

\begin{proposition}\label{prop:restric-compat}
  Restriction is compatible with observational equivalence.
\end{proposition}
\begin{proof}
  First observe that, since $Y$ is supposed to be closed by permutations,
  restriction commutes with permutations, hence
  $\psplit(\restric{Y}{u})=\restric{Y}{\psplit(u)}$.
  Then by the representation property (Proposition~\ref{prop:represent}), if
  $u\obseqv v$ then $\psplit(u)\obseqv[\web{X}]\psplit(v)$, so it suffices to
  prove that restriction is compatible with observational equivalence in
  static arenas.

  Let $Z$ be a finite subset of $\web{X}∖ Y$, define $\extend{Z}$ as
  the linear map such that for all $t∈\Plays{X}$, $\extend{Z}{t}$ is the
  play on $\web{t}\cup Z$, whose preorder relation is $\order{t}$ extended as
  the identity relation on $Z$.
  Let $r$ be a play in $\Plays{X}$ such that $\web{r}∖ Y=Z$.
  Suppose $r$ is acyclic, then for all play $s$ with $\web{s}\cup Z=\web{r}$
  the relation $\order{r}\cup\order{s}$ is acyclic if and only if
  $(\order{r}\cap Y^2)\cup\order{s}$ is acyclic, so we have
  $\outcome{r\sync\extend{Z}{s}}=\outcome{\restric{Y}{r}\sync s}$.
  If $r$ is not acyclic, then the equality holds too since both sides are $0$.
  The equality extends trivially to all plays $s$ such that $\web{s}\subset Y$.

  Consider a pair $u\obseqv[\web{X}]v$ and a play $s∈\Plays{X}$ with
  $\web{s}\subset Y$.
  By Proposition~\ref{prop:direct-sum} we can decompose $u$ as
  $\sum_{C∈\PartsFin{\web{X}}}u_C$ with $u_C∈\PreOrdAlgS{C}$ and
  similarly for $v$ so that for each $C$ we have $u_C\obseqv[\web{X}]v_C$.
  For a given $C$, let $Z=C∖ Y$, then by linearity of the equation in
  the previous paragraph we get
  $\outcome{\restric{Y}{u_Z}\sync s}=\outcome{u_Z\sync\extend{Z}{s}}$ for all
  play $s$ with $\web{s}\subset Y$, and by the equivalence of $u_Z$ and $v_Z$
  we get
  $\outcome{\restric{Y}{u_Z}\sync s}=\outcome{\restric{Y}{v_Z}\sync s}$.
  This trivially holds too if $\web{s}\not\subset Z$, so we get the equivalence
  $\restric{Y}{u_Z}\obseqv\restric{Y}{v_Z}$, and we deduce
  $\restric{Y}{u}\obseqv\restric{Y}{v}$ by linearity.
\end{proof}

\begin{definition}
  Let $X,Y,Z$ be three arenas with pairwise disjoint supports.
  Define \emph{partial static synchronisation} along $X$ as the bilinear map
  $\sync[X]$ from $\PreOrdAlg{X+Y}×\PreOrdAlg{X+Z}$ to
  $\PreOrdAlg{X+Y+Z}$ such that for all $r∈\Plays{X+Y}$ and
  $s∈\Plays{X+Z}$,
  \[
    r \sync[X] s :=
    \begin{cases}
      \bigl( \web{r}\cup\web{s}, (\order{r}\cup\order{s})^* \bigr)
        &\text{if } \web{r}\cap\web{X} = \web{s}\cap\web{X} \\
      0 &\text{otherwise}
    \end{cases}
  \]
  Deduce partial permuted synchronisation as
  \[
    r \psync[X] s := \multiplicity[X]{\restric{X}{s}}
      \sum_{s'∈\porbit[X]{s}} r \sync[X] s'
  \]
\end{definition}

\begin{example}
  Consider an arena $X$ containing at least two interchangeable actions
  labelled $a_1,a_2$ and arenas $Y$ and $Z$ containing events $b$ and $c$
  respectively.
  Then we have the partial static synchronisation
  \[
    \tikzformula{
      \eventnode{(0,0)}{b}{left}{b}
      \eventnode{(0,1)}{a1}{left}{a_1}
      \eventnode{(1,0)}{a2}{right}{a_2}
      \draw (b) -- (a1); }
    \sync[X]
    \tikzformula{
      \eventnode{(0,0)}{a1}{left}{a1}
      \eventnode{(0,1)}{c}{left}{c}
      \eventnode{(1,0)}{a2}{right}{a_2}
      \draw (a1) -- (c); }
    = \tikzformula{
      \eventnode{(0,0)}{b}{left}{b}
      \eventnode{(0,1)}{a1}{left}{a1}
      \eventnode{(0,2)}{c}{left}{c}
      \eventnode{(1,0)}{a2}{right}{a_2}
      \draw (b) -- (c); }
  \]
  and the partial permuted synchronisation
  \[
    \tikzformula{
      \eventnode{(0,0)}{b}{left}{b}
      \eventnode{(0,1)}{a1}{left}{a_1}
      \eventnode{(1,0)}{a2}{right}{a_2}
      \draw (b) -- (a1); }
    \psync[X]
    \tikzformula{
      \eventnode{(0,0)}{a1}{left}{a1}
      \eventnode{(0,1)}{c}{left}{c}
      \eventnode{(1,0)}{a2}{right}{a_2}
      \draw (a1) -- (c); }
    = \tikzformula{
      \eventnode{(0,0)}{b}{left}{b}
      \eventnode{(0,1)}{a1}{left}{a1}
      \eventnode{(0,2)}{c}{left}{c}
      \eventnode{(1,0)}{a2}{right}{a_2}
      \draw (b) -- (c); }
    + \tikzformula{
      \eventnode{(0,0)}{b}{left}{b}
      \eventnode{(0,1)}{a1}{left}{a1}
      \eventnode{(1,0)}{a2}{right}{a_2}
      \eventnode{(1,1)}{c}{right}{c}
      \draw (b) -- (a1) (a2) -- (c); } .
  \]
  The factor $\multiplicity[X]{\restric{X}{s}}$ (which is $1$ in this example)
  plays the same role as in full synchronisation (Definition~\ref{def:psync}),
  remarking that we only apply permutations on the $X$ part.
\end{example}

\begin{proposition}
  Let $X,Y,Z$ be three arenas with pairwise disjoint supports.
  Partial synchronisation along $X$ is associative as
  \[
    (u\psync[X]v)\psync[X+Y+Z]w
    = u\psync[X+Y](v\psync[X+Z]w)
  \]
  for all $u∈\PreOrdAlg{X+Y}$, $v∈\PreOrdAlg{X+Z}$ and
  $w∈\PreOrdAlg{X+Y+Z}$.
  It is compatible with observational equivalence and commutative up to
  equivalence.
\end{proposition}
\begin{proof}
  Let $X,Y,Z$ be three disjoint arenas.
  Consider three plays $r∈\Plays{X+Y}$, $s∈\Plays{X+Z}$ and
  $t∈\Plays{X+Y+Z}$.
  The partial synchronisations $(r\sync[X]s)\sync[X+Y+Z]t$ and
  $r\sync[X+Y](s\sync[X+Z]t)$ are non-zero if and only if we can define
  \[
    A=\web{r}\cap\web{X}=\web{s}\cap\web{X}=\web{t}\cap\web{X}, \quad
    B=\web{r}\cap\web{Y}=\web{t}\cap\web{Y}, \quad
    C=\web{s}\cap\web{Z}=\web{t}\cap\web{Z},
  \]
  and in this case the result is the play on $A\cup B\cup C$ whose preorder
  relation is $(\order{r}\cup\order{s}\cup\order{t})^*$, so we have
  \[
    (r\sync[X]s)\sync[X+Y+Z]t = r\sync[X+Y](s\sync[X+Z]t) .
  \]
  Assume representants are chosen in each arena in such a way that for
  $D\subset\web{X}$ and $E\subset\web{Y}$,
  $\repr{D\cup E}=\repr{D}\cup\repr{E}$, and similarly for $X+Z$ and $X+Y+Z$.
  Choosing representants this way is always possible since permutations of
  $X$, $Y$ and $Z$ are independent in the sum arenas.
  Suppose $r,s,t$ and $A,B,C$ are representants.
  $\rpgroup{A}$ is the same in all sums of arenas that involve $X$, and
  similarly for $\rpgroup{B}$ and $\rpgroup{C}$.
  Moreover we have
  $\rpgroup{A\cup B\cup C}=\rpgroup{A}×\rpgroup{B}×\rpgroup{C}$, and
  similarly for other unions, so by similar considerations as for permuted
  synchronisation, we get
  \begin{multline*}
    (r\psync[X]s)\psync[X+Y+Z]t
    = \sum (r\sync[X]σ_As)\sync[X+Y+Z]σ'_Aσ_Bσ_Ct
    = \sum r\sync[X+Y](σ_As\sync[X+Z]σ'_Aσ_Bσ_Ct) \\
    = \sum r\sync[X+Y]σ_Aσ_B(s\sync[X+Z]σ'_Aσ_Ct)
    = r\psync[X+Y](s\psync[X+Z]t)
  \end{multline*}
  where the sums are indexed on 
  $(σ_A,σ'_A,σ_B,σ_C)
  ∈\rpgroup{A}×\rpgroup{A}×\rpgroup{B}×\rpgroup{C}$.
  Partial synchronisation commutes with permutations so this equality extends
  to plays that are not representants, and by linearity is extends to
  arbitrary vectors.

  Now consider $u,u'∈\PreOrdAlg{X+Y}$, $v∈\PreOrdAlg{X+Z}$ and
  $w∈\PreOrdAlg{X+Y+Z}$, and suppose $u\obseqv u'$.
  By the same arguments as in the proof of
  Proposition~\ref{prop:restric-compat}, we get the equality
  $\outcome{u\psync[X+Y](v\psync[X+Z]w)}
  =\outcome{u\psync\restric{X+Y}(v\psync[X+Z]w)}$, then $u\obseqv u'$ implies
  that these are equal to $\outcome{u'\psync[X+Y](v\psync[X+Z]w)}$, and
  applying associativity on this we can deduce
  $u\psync[X]v\obseqv u'\psync[X]v$. Commutativity of partial synchronisation
  is obvious, and it yields the compatibility with observational equivalence
  on the right.
\end{proof}

By similar arguments, we prove other ``localized'' associativities, the
general case being
\[
  ( u \psync[A+B]     v ) \psync[A+C+D] w
  = u \psync[A+B+C] ( v   \psync[A+D]   w )
\]
for $u∈\OrdAlg{A+B+C+E}$, $v∈\OrdAlg{A+B+D+F}$ and $w∈\OrdAlg{A+C+D+G}$,
where $A,B,C,D,E,F,G$ are seven (!) pairwise disjoint arenas.
Although this formulation is frighteningly heavy, the point is rather simple:
when partially synchronising two vectors $u$ and $v$, synchronise them along
the arenas they have in common, and the result will be on the union of the
arenas of $u$ and $v$.


The simplest case of partial synchronisation is when ``synchronising'' two
vectors $u∈\OrdAlg{X}$ and $v∈\OrdAlg{Y}$ along the empty arena, yielding
$u\psync[\emptyset]v∈\OrdAlg{X+Y}$.
In this case, $u$ and $v$ are essentially kept independent, which in
particular implies
\[
  \outcome{u\psync[\emptyset]v} = \outcome{u}\outcome{v}.
\]
This is deduced by linearity from the case of plays, remarking that for
$r∈\Plays{X}$ and $s∈\Plays{Y}$, $r\psync[\emptyset]s$ is the disjoint
union of $r$ and $s$, which is consistent if and only if $r$ and $s$ are
consistent.

\begin{definition}
  Let $X$ and $Y$ be two arenas.
  Define the bilinear map $\otimes$ from $\OrdAlg{X}×\OrdAlg{Y}$ to
  $\OrdAlg{X+Y}$ as $u\otimes v:=u\psync[\emptyset]v$.
  For two types $A:X$ and $B:Y$, define $A\otimes B$ as the submodule of
  $\OrdAlg{X+Y}$ generated by the image of $A×B$ by $\otimes$.
\end{definition}

Simply put, $A\otimes B$ is the $\Scal$-module consisting of processes that
can be written as juxtapositions of a process in $A$ and a process in $B$
with no scheduling constraint between them, or as a sums of such things,
\emph{up to observational equivalence}.
As illustrated in Example~\ref{ex:type}, this does not imply that any vector
$\sum_{i∈I}λ_ir_i∈A\otimes B$ is syntactically a sum of $u_i\psync[∅]v_i$ with
$u_i∈A$ and $v_i∈B$.

\begin{proposition}
  If $\Scal$ is a rational ring, then for all types $A:X$ and $B:Y$,
  $A\otimes B$ is the tensor product of $A$ and $B$ in the sense of
  $\Scal$-algebras.
\end{proposition}
\begin{proof}
  By Proposition~\ref{prop:type-basis} the types $A$ and $B$ have bases
  $(b_i)_{i∈I}$ and $(c_j)_{j∈J}$, and there are families of vectors
  $(b^*_i)_{i∈I}$ and $(c^*_j)_{j∈J}$ such that each $b^*_n$ identifies
  $b_n$ among the elements of $(b_i)_{i∈I}$, and similarly for $c^*_n$.
  We prove that the vectors $b_i\psync[\emptyset]c_j$ are linearly
  independent.
  Consider a linear combination
  $u=\sum_{(i,j)∈I× J}λ_{ij}(b_i\psync c_j)$ in $\OrdAlg{X+Y}$.
  For each $(m,n)∈I× J$ we have
  \begin{multline*}
    \outcome{ u \psync[X+Y] (b^*_m \psync[\emptyset] c^*_n) }
    = \sum_{(i,j)∈I× J} λ_{mn}
      \outcome{ (b_{i} \psync[\emptyset] c_{j})
        \psync[X+Y] (b^*_m \psync[\emptyset] c^*_n) } \\
    = \sum_{(i,j)∈I× J} λ_{mn}
      \outcome{ (b_{i} \psync[X] b^*_m)
        \psync[\emptyset] (c_{j} \psync[Y] c^*_n) }
    = \sum_{(i,j)∈I× J} λ_{mn}
      \outcome{ b_{i} \psync[X] b^*_m }
      \outcome{ c_{j} \psync[Y] c^*_n }
  \end{multline*}
  using the associativity properties stated above.
  By definition of $b^*_m$ and $c^*_n$, the only non-zero term in the final
  sum is for $(i,j)=(m,n)$, and this term is $λ_{mn}$.
  Applying this on every $(m,n)$ implies the unicity of the decomposition of
  $u$ on the $b_m\otimes c_n$.
  So the $b_m\otimes c_n$ form a linearly independent family, which proves
  that $A\otimes B$ is isomorphic to the tensor product of $A$ and $B$, as
  $\Scal$-modules. The associativity property ensures that they are also
  isomorphic as
  $\Scal$-algebras.
\end{proof}

\begin{definition}
  Let $X$, $Y$ and $Z$ be three arenas.
  Let $u∈\OrdAlg{X+Y}$ and $v∈\OrdAlg{Y+Z}$.
  Composition of $u$ and $v$ through $Y$ is the vector
  $u\compo[Y]v:=\restric{X+Z}{(u\psync[Y]v)}∈\OrdAlg{X+Z}$.

  Let $A:X$ and $B:Y$ be two types.
  The type $A\multimap B:X+Y$ is the submodule of $\OrdAlg{X+Y}$ generated by
  all plays $r$ such that for all $u∈A$, $r\compo[X]u∈B$.
\end{definition}

By the remarks above, we get associativity of composition.
In the special case where $X$ is the empty arena, $\OrdAlg{X+Y}$ is equal to
$\OrdAlg{Y}$ and $u\compo[Y]v$ is a vector in $Z$, so $v$ induces a linear map
from $\OrdAlg{Y}$ to $\OrdAlg{Z}$.
However, this mapping from vectors of $\OrdAlg{X+Y}$ to linear
maps from $\OrdAlg{X}$ to $\OrdAlg{Y}$ is neither injective nor surjective.

It is easy to check the standard adjunction
$A\multimap(B\multimap C)=(A\otimes B)\multimap C$ for all types $A,B,C$ of
pairwise disjoint supports.
Moreover, if we call $\one$ the non-trivial type over the empty arena, which
is isomorphic to $\Scal$, we have for all type $A$ that
$A\otimes\one=\one\otimes A=\one\multimap A=A$.

\subsection{Bialgebraic structure} 

\begin{definition}\label{def:indexing}
  Let $X$ and $Y$ be two arenas.
  Define the \emph{indexing} of $Y$ by $X$ as the arena
  \[
    X\indexing Y := \bigl( \web{X}×\web{Y},
      \pgroup[X]×(\pgroup[Y])^{\web{X}} \bigr)
  \]
  where permutations act as
  \[
    (σ,φ)(x,y) := (σx, φ(x)\,y)
  \]
\end{definition}

We interpret indexing as follows: $\web{X\indexing Y}$ consists of copies of
$Y$ indexed by points of $X$.
A permutation in $X\indexing Y$ consists in permuting each copy independently,
using the function $φ:\web{X}\to\pgroup[Y]$ that provides a permutation for
each copy, and then permuting the copies themselves using a permutation in
$X$.

Note that we easily get the equality
$(X+Y)\indexing Z=(X\indexing Z)+(Y\indexing Z)$,
however $X\indexing(Y+Z)$ is not equal to $(X\indexing Y)+(X\indexing Z)$,
since permutations of copies in the former operate the same way on the copies
of $X$ and those of $Y$, while in the latter they may not.
There is also an isomorphism between $(X\indexing Y)\indexing Z$ and
$X\indexing(Y\indexing Z)$, and these appear as
$(\web{X}×\web{Y}×\web{Z},
\pgroup[X]×(\pgroup[Y])^{\web{X}}×(\pgroup[Z])^{\web{X}×\web{Y}})$.

The structure of the indexing arena is used only for identifying and permuting
copies, in particular we will not consider plays on this arena.
The primary purpose of indexing is to build an arena in which the symmetric
algebra over a given type will fit.
It also generalises the direct sum when the indexing arena is static.

\begin{definition}
  Let $\PNat$ be the arena with $\web{\PNat}=\Nat$, the set of natural
  numbers, and $\pgroup[\PNat]=\Perm{\Nat}$, the group of all permutations of
  $\Nat$.
  For all arena $X$, define $\sharp X:=\PNat\indexing X$.
\end{definition}

So the arena $\PNat\indexing X$ contains a countable number of interchangeable
copies of $X$.
Another useful construct is the following:
identifying each integer $n$ with the set
$\{0,\ldots,n-1\}$, which is in turn identified with the static arena with
this set as the web, the arena $n\indexing X$ is  isomorphic to the sum
$X+\cdots+X$ with $n$ independent copies of $X$.
If $n=0$, this yields the empty arena $\emptyset$.
Then $n\indexing\sharp X=n\indexing\PNat\indexing X$ contains a countable set
of copies of $X$, partitioned into $n$ countable classes of interchangeable
copies.

\begin{definition}\label{def:gamma-delta}
  Let $n$ be a strictly positive integer, let $φ$ be a bijection from
  $n×\Nat$ to $\Nat$.
  For all arena $X$, define the function
  $γ^n_φ:\Plays{n\indexing\sharp X}\to\Plays{\sharp X}$ as
  \[
    \web{γ^n_φs} := \set[2]{ (φ(i),x) }{ (i,x) ∈\web{s} }
    \quad\text{and}\quad
    (φ(i),x) ≤_{γ^n_φ{s}} (φ(j),y)
      \text{ iff } (i,x) ≤_s (j,y) .
  \]
  Define the linear map
  $δ^n:\PreOrdAlg{\sharp X}\to\PreOrdAlg{n\indexing\sharp X}$ as
  \[
    δ^ns := \sum_{c\,:\,\pi_1(\web{s})\to n} c\bullet s
    \quad\text{with}\quad
    \begin{array}{l}
      \web{c\bullet s} := \set[2]{ ((c(i),i),x) }{ (i,x)∈\web{s}, i∈A } \\
      ((c(i),i),x) ≤_{c\bullet s} ((c(j),j),y) \text{ iff } (i,x) ≤_s (j,y)
    \end{array}
  \]
  where $\pi_1$ is the first projection, so
  $\pi_1(\web{s})=\set{i}{(i,x)∈\web{s}}$.
\end{definition}

The function $γ^n_φ$ is a simple renaming of the copies of $X$ using the
function $φ$, which extends the bijection $φ:n×\Nat\to\Nat$ to a
bijection between $\Plays{n\indexing\sharp X}$ and $\Plays{\sharp X}$.
As explained below, this bijection is compatible with observational
equivalence, but its quotient is not injective.
Instead, it fuses the $n$ independent copies of $\sharp X$ into one, which
makes events from different copies interchangeable.

The linear map $δ^n$ acts as a non-deterministic inverse operation.
Given a play $s$ in $\Plays{\sharp X}$, it enumerates all possible ways of
partitioning the events of $s$ into $n$ identified subsets.
The function $c$ represents such a choice, and $c\bullet r$ applies this
choice to the play $r$.

As we shall see, the operators $γ^n_φ$ and $δ^n$ are very similar to a
multiplication and comultiplication in a bialgebra.
They are analogous to concatenation and deconcatenation, which give a
bialgebraic structure to tensor algebras~\cite{lod08:triples}.

\begin{proposition}\label{prop:gamma-delta}
  Let $X$ be an arena and let $n$ be a strictly positive integer.
  The maps $γ^n_φ$ and $δ$ are compatible with observational
  equivalence and the quotient map of $γ^n_φ$ is independent of $φ$.
  For all vectors $u∈\OrdAlg{n\indexing\sharp X}$ and $v∈\OrdAlg{\sharp X}$
  we have
  \[
    γ^n(u) \psync[\sharp X] v \obseqv[\sharp X]
    u \psync[n\indexing\sharp X] δ^n(v) .
  \]
\end{proposition}
\begin{proof}
  Observe that for any permutation
  $σ∈\pgroup[n\indexing\Nat\indexing\sharp X]$ (\ie\ a family of
  independent permutations on each copy of $X$ in $n\indexing\sharp X$)
  there is a permutation $σ'∈\pgroup[\Nat\indexing\sharp X]$ such that
  $γ^n_φ\circσ=σ'\circγ^n_φ$, and the other way
  around for $δ^n$.
  As a consequence, by Proposition~\ref{prop:represent}, we can deduce the
  expected result from the case where $X$ is static.
  Then all considered permutations are in $\Perm{n×\Nat}$ and~$\Perm{\Nat}$.

  The map $γ^n_φ$ decomposes as the injection of
  $\PreOrdAlg{(n\indexing\PNat)\indexing X}$ into
  $\PreOrdAlg{\Perm{n×\Nat}\indexing X}$ and the renaming of
  $\Perm{n×\Nat}\indexing X$ into $\PNat\indexing X$ through $φ$.
  The former consists in growing the permutation group on a fixed web and the
  latter is an isomorphism, so both are compatible with observational
  equivalence.
  For $δ^n$, given a permutation $σ∈\Perm{\Nat}$ and a
  play $s$, for all choice function $c$ for $σ s$ we have
  $c\bulletσ s=σ'(c'\bullet s)$ with $c'=c\circσ$ and
  $σ'(i,j)=(i,σ(j))$, which establishes a bijection between the
  choices of $s$ and those of $σ s$.
  From this we can conclude that $δ^n$ is compatible with observational
  equivalence.

  Let $r∈\Plays{n\indexing\sharp X}$, let $s∈\Plays{\sharp X}$ and let
  $φ$ be a bijection from $n×\Nat$ to $\Nat$.
  Assume $s$ is a representant.
  First suppose $\repr{\web{γ^n_φ(r)}}\neq\web{s}$,
  then $γ^n(r)\psync s$ is zero.
  Suppose that there is a choice $c$ such that
  $\repr{\web{c\bullet s}}=\web{r}$, then we get a permutation
  $σ∈\pgroup[n\indexing\sharp X]$ that induces a bijection from
  $\web{c\bullet s}$ to $\web{r}$.
  By definition $σ$ is a bijection between the pairs $(c(i),i)$ and
  $\pi_1(\web{r})$, which can be extended into a bijection $\psi$ from
  $n×\Nat$ to $\Nat$, such that
  $\repr{\web{γ^n_φ(r)}}=\web{s}$.
  This contradicts the hypothesis $\repr{\web{γ^n_φ(r)}}\neq\web{s}$
  since $γ^n_φ(r)$ and $γ^n_\psi(r)$ are necessarily permutations
  of each other, from the remarks above.
  Hence for all $c$ we have $\repr{\web{c\bullet s}}\neq\web{r}$, so
  $r\psyncδ^n(s)=0$, and the equality holds.

  Now suppose $\repr{\web{γ^n_φ(r)}}=\web{s}$.
  Applying a suitable permutation to $r$ and choosing $φ$ appropriately
  (we know from the above remarks that these operations are allowed) we can
  assume that $γ^n_φ(r)$ is a representant, so
  $\web{γ^n_φ(r)}=\web{s}$, and
  $γ^n_φ(r)\psync s=\sum_{σ∈G}φ r\syncσ s
  \obseqv\sum_{σ∈G}r\syncφ^{-1}σ s$,
  where $G$ is the group of permutations of $\web{s}$ induced by
  $\pgroup[\sharp X]$, that is the symmetric group of $\pi_1(\web{s})$.
  For each $σ∈G$, the function
  $\pi_1φ^{-1}σ$ is a choice function $c_σ$ over
  $\pi_1(\web{s})$, and $σ^*(i,j):=(i,σ^{-1}φ(i,j))$
  is a permutation in $\pgroup[n\indexing\Nat]$ such that
  $σ^*φ^{-1}σ(i)=(c_\tau(i),i)$, hence
  $σ^*φ^{-1}σ(s)=c_σ\bullet s$.
  By partitioning the sum for $\tau∈G$ according to choice functions, we
  get $γ^n_φ(r)\psync s \obseqv
  \sum_c\sum_{σ∈G,c_σ=c}r\syncσ^{*{-1}}(c_σ\bullet s)$.
  By construction, for a fixed $c$, we have
  $\set{σ'^{*{-1}}}{σ'∈G,c_{σ'}=c}
  =\pgroup[n\indexing\Nat]σ^{*{-1}}$, so we get
  $γ^n_φ(r)\psync s\obseqv[\sharp X]
  \sum_cr\psync[n\indexing\sharp X](c\bullet s)
  =r\psync[n\indexing\sharp X]δ^n(s)$,
  from which we conclude by linearity.
\end{proof}

As a consequence, $\OrdAlg{\sharp X}$ has the structure of a commutative
algebra with $γ^2$ as the multiplication and the empty play as the unit.
The $δ^2$ does not make it a bialgebra in general, because for an
arbitrary $u∈\OrdAlg{\sharp X}$, $δ^2(u)∈\OrdAlg{\sharp X+\sharp X}$
has no reason to be in the tensor product
$\OrdAlg{\sharp X}\otimes\OrdAlg{\sharp X}$.
The reason is that a given play in $\Plays{\sharp X}$, the components of
$δ^2(r)$ are not disjoint unions of plays on the two copies of
$\sharp X$, but they may contain scheduling constraints that involve both
copies.
We do get a bialgebra if we restrict to the case of plays in which all copies
stay independent.

\begin{definition}
  Let $X$ be an arena.
  For all integer $n$ and play $r∈\Plays{X}$, define the play
  $n\bullet r∈\Plays{\sharp X}$ as in Definition~\ref{def:gamma-delta} for
  the constant function $n$.
  This obviously induces an isomorphism between $\OrdAlg{X}$ and
  $\OrdAlg{\{n\}\indexing X}$, which maps each type $A:X$ to an isomorphic
  type $n\bullet A:\{n\}\indexing X$.
  However, the $\{n\}\indexing X$ for distinct $n$ are disjoint.

  The arena $\{n\}\indexing X$ is included in $\sharp X$, let
  $ε_n:\PreOrdAlg{X}\to\PreOrdAlg{\sharp X}$ be the inclusion map.
  Clearly all the $ε_n$ are compatible with observational equivalence
  and their quotients are all equal.
  Name $ε:\OrdAlg{X}\to\OrdAlg{\sharp X}$ the quotient map.

  For all type $A:X$, define the type $\oc A:\sharp X$ as
  \[
    \oc A := \sum_{n∈\Nat} γ^n(A^n)
    \quad\text{where}\quad
    (A^n : n\indexing X) := \bigotimes_{i=0}^{n-1} (i\bulletε(A))
  \]
  Define the degree of a vector $u∈\oc A$ as the smallest integer $d(u)$ such
  that $u$ is in the partial sum $\sum_{n≤ d(u)}γ^n(A^n)$.
\end{definition}

For all type $A:X$, the type $\oc A:\sharp X$ is again a commutative algebra
with $γ^2$ as the product and the empty play as the unit.
The degree function makes it a graded algebra, intuitively the degree of a
vector $u$ is the maximum number of different copies of $A$ that $u$ uses.
If the type $A$ is \emph{strict} (\ie\ if it does not contain the empty play)
and $\Scal$ is rational, then $\oc A$ is isomorphic to the symmetric
algebra of $A$.
The strictness condition means that each copy of $A$ is actually used, without
this hypothesis the isomorphism fails because all powers of the empty play are
necessarily equal to the empty play in $\oc A$.

The linear map $δ^2$ also makes $\oc A$ a cocommutative coalgebra whose
counit is the linear form that maps the empty play to $1$ and non-empty plays
to $0$.
It is routine to check that the algebra and coalgebra structure are
compatible, making $\oc A$ a bialgebra.
Interestingly, if $A$ is the unique strict type on the singleton arena (which
is isomorphic to $\Scal$), then $\oc A$ is isomorphic to the bialgebra of
polynomials in one variable over $\Scal$.

\subsection{Towards differential linear logic} 

The mapping $A\mapsto\oc A$ is a functor in the category of types and linear
maps.
The map $ε$ from the definition above is a natural transformation from
$A$ to $\oc A$, and by choosing a bijection from $\Nat×\Nat$
into $\Nat$ we get a natural transformation from $\oc\oc A$ to $\oc A$ which
makes $\oc A$ into a monad (the choice of a particular bijection is
unimportant, for the same reasons as in Proposition~\ref{prop:gamma-delta}).
The quotient of the linear map that sends each play $n\bullet r$ to $r$ and
all other plays to zero is a natural transformation from $\oc A$ to $A$, and
using any bijection from $\Nat$ to $\Nat×\Nat$ we get a natural
transformation from $\oc A$ to $\oc\oc A$, which also makes $A$ a comonad.

We can also check the isomorphism $\oc(A⊕ B)≃\oc A\otimes\oc B$
for any strict types $A:X$ and $B:Y$ over disjoint arenas.
The first type is in the arena $\PNat\indexing(X+Y)$ and the second one is in
$(\PNat\indexing X)+(\PNat\indexing Y)$; these arenas are not isomorphic but
the types themselves are thanks to the definition of the direct sum.

All these considerations show that the structure of our types supports most
constructs of differential linear logic~\cite{er06:diffnets}, including
additives, multiplicatives and exponentials with structural and costructural
rules.
However, the construction is not yet a model of differential logic, for
several reasons:
\begin{itemize}
\item 
  One crucial thing that lacks in our framework is the axioms.
  They do not fit in the present work because our objects are too
  finitary: all vectors are finite linear combinations of finite plays, hence
  there can be no vector in $A\multimap A$ that is neutral for composition as
  soon as $A$ is not finite dimensional.
  The reason is similar to the case of units for synchronisation in
  Proposition~\ref{prop:units}, and solving this defect requires a radical
  extension of this work, as explained in the introduction.
\item 
  The proper notion of duality needed to interpret logic, or equivalently the
  definition of the type $\bot$, is not clear at first sight.
  This type must be defined on the empty arena, and our notion of type only
  leaves two choices: $\one=\OrdAlg{\emptyset}$ and $\{0\}$.
  The first one is degenerate given our definition of $A\multimap B$, the
  second one yields orthogonality with respect to the bilinear form
  $(u,v)\mapsto\outcome{u\psync v}$ (note however that this bilinear form is
  not a scalar product, because it is not positive).
  We will not explore this case here because it exceeds the scope of the
  present work.
\item 
  Of course, building a model of linear logic requires to prove that the
  interpretation of proofs is preserved by cut-elimination.
  Most tools are present for this, assuming we restrict to an ill-structured
  logical system without the axiom rule.
  Here again, we defer this task to further work, as the questions of axioms
  and duality obviously have to be answered first for this to be of interest.
\end{itemize}

\section{Interpretation of process calculi} 
\label{sec:processes}

In this section, we detail how process calculi can be interpreted in order
algebras.
As a particular case to work on, we use the \pii-calculus with internal
mobility~\cite{sg96:pii}, that is the fragment of the \pii-calculus where
output actions can only send fresh names.
Most development here could be carried out in other similar calculi.
Had we used CCS, essentially everything would have been the same up to
section~\ref{sec-traces}, in which the definition of arenas would have been
simpler because of the mostly trivial name structure of CCS.
The full \pii-calculus, on the other hand, would have required the handling of
equality tests between names, which is perfectly doable at the cost of
trickier definitions; this exceeds the scope of the present work.

\subsection{Quantitative testing} 
\label{sec-testing}

We consider the \pii-calculus with internal mobility, or
\piI-calculus, extended with \emph{outcomes} from a commutative
semiring~$\Scal$.
We consider the monadic variant of the calculus for simplicity, but using the
polyadic form would not pose any significant problem.
More importantly, we restrict to finite processes.
\begin{definition}\label{def-terms}
  We assume a countable set $\Names$ of names.
  Polarities are elements of $\Pola=\implem{\shpos,\shneg}$.
  Terms are generated by the following grammar:
  \begin{syntax}
  \define[branchings] S,T
    \case \loc{ι}{u}^ε(x).P
      \comment{action, with $u,x∈\Names$, $ε∈\Pola$ and $ι∈\Nat$}
    \case S+T  \comment{external choice}
  \define[processes] P, Q
    \case λ  \comment{outcome, with $λ∈\Scal$}
    \case S  \comment{branching}
    \case P\para Q  \comment{parallel composition}
    \case\new{x}P  \comment{hiding, with $x∈\Names$}
  \end{syntax}
  In an action $\loc{ι}{u}^ε(x).P$, $ι$ is the \emph{location}, $u$ is
  the \emph{subject}, $x$ is the \emph{object} and $P$ is the
  \emph{continuation}.
  The name $x$ is bound in $P$ by the action, independently of the
  polarity~$ε$.

  Terms are considered up to injective renaming of bound names and commutation
  of restrictions, i.e. $\new{x}\new{y}P=\new{y}\new{x}P$, with the standard
  convention that all bound names are distinct from all other names.
  We also impose that in a given term all locations are always distinct.
  The set of locations occurring in a term $P$ is written $\web{P}$.
\end{definition}

Actions (without continuations) will be ranged over by Greek letters $α,β$, so
that we can write expressions like $α.P$ or $α.(β.Q\para R)$.
By convention, an action $u^{\shpos}(x)$ is called \emph{positive} and is also
written $u(x)$, an action $u^{\shneg}(x)$ is called \emph{negative} and is
also written $\bar{u}(x)$.
More generally, if $α$ is an action, we write $\bar{α}$ for the action with
the same subject and the opposite polarity, in particular $\bar{u}^ε(x)$ is
the action of the opposite polarity as $u^ε(x)$.

Locations are simply a way to give different identities to different
occurrences of a given channel name in a term, so we can talk about ``the
action $ι$'' in an unambiguous manner.
Renamings of these locations are of course unobservable by the processes, so
the distinctness condition is not a restrictions on the terms we can write.
Terms with locations can be seen as decorations on standard terms of the
\piI-calculus.

We want to define an operational semantics in which commutation of independent
transitions is allowed.
To make this possible by only looking at transition labels, we enrich the
labels using locations so that different occurrences of a given action are
distinguishable at the level of operational semantics.

\begin{definition}
  Transition labels can be of one of two kinds:
  \begin{syntax}
  \define a, b
    \case \vtlabel{u^ε(x)}{ι} \comment{visible action}
    \case \itlabel{ι}{κ}  \comment{internal transition}
  \end{syntax}
  Transitions are derived by the rules of Table~\ref{table:dlts}.
  The notation $x∉a$ means that the name $x$ does not occur (free or
  bound) in the label $a$.

  An interaction is a finite sequence of transition labels.
  A path is a finite sequence of internal transition labels.
  An interaction $p=a_1a_2\ldots a_n$ is valid for $P$, written $p∈P$, if
  there are valid transitions
  $P\trans{a_1}P_1\trans{a_2}\cdots\trans{a_n}P_n$.
\end{definition}
\begin{table}
  \centering
  \begin{prooftree}
    \Infer0{ \loc{ι}{α}.P \vtrans{α}{ι} P }
  \end{prooftree}
  \hfil
  \begin{prooftree}
    \Hypo{ P \vtrans{u^ε(x)}{ι} P' }
    \Hypo{ Q \vtrans{\bar{u}^ε(y)}{κ} Q' }
    \Infer2{ P\para Q \itrans{ι}{κ} \new{x}(P'\para Q'[x/y]) }
  \end{prooftree}
  \hfil
  \begin{prooftree}
    \Hypo{ P \trans{a} P' }
    \Hypo{ x ∉ a }
    \Infer2{ \new{x}P \trans{a} \new{x}P' }
  \end{prooftree}
  \\[1.5ex]
  \begin{prooftree}
    \Hypo{ S \trans{a} S' }
    \Infer1{ S+T \trans{a} S' }
  \end{prooftree}
  \hfil
  \begin{prooftree}
    \Hypo{ S \trans{a} S' }
    \Infer1{ T+S \trans{a} S' }
  \end{prooftree}
  \hfil
  \begin{prooftree}
    \Hypo{ P \trans{a} P' }
    \Infer1{ P\para Q \trans{a} P'\para Q }
  \end{prooftree}
  \hfil
  \begin{prooftree}
    \Hypo{ P \trans{a} P' }
    \Infer1{ Q\para P \trans{a} Q\para P' }
  \end{prooftree}
  \smallskip
  \caption{Decorated labelled transition system for the \piI-calculus}
  \label{table:dlts}
\end{table}

The use of decorations to define a parallel operational semantics was first
proposed by Boudol and Castellani as ``proved
transitions''~\cite{bc88:ccs,dd89:causal}, and the technique we use
here can be seen as a simplification for our purpose.
It is clear that for all term $P$ and interaction $p$, there is at most one
term $P/p$ (exactly one if $p∈P$) such that there is a transition sequence
$P\trans{p}P/p$ (up to renaming of revealed bound names).
Note that by removing all locations from labels (replacing $\itlabel{ι}{κ}$
by $τ$) one gets the standard labeled transition system for the \piI-calculus.
For this reason, we allow ourselves to keep locations implicit when they are
not important.

\begin{definition}\label{def-homotopy}
  Prefixing in a term $P$ is the partial order $≤_P$ over $\web{P}$ such that
  $ι<_Pκ$ when in $P$ the action at location $κ$ occurs in the continuation of
  the action at location $ι$.
  Two labels $a$ and $b$ are \emph{independent}, written $a\indep_Pb$, if all
  locations occurring in $a$ or $b$ are distinct and pairwise incomparable for
  prefixing.
  Homotopy in a term $P$ is the smallest equivalence $\homo_P$ over
  interactions of $P$ such that $pabq\homo_Ppbaq$ when $a\indep_Pb$.
\end{definition}

Two execution paths of a given term are homotopic if it is possible to
transform one into the other by exchanging consecutive transitions if they are
independent.
Prefixing generates local constraints which propagate to paths by this
relation.

\begin{proposition}\label{prop-permutation}
  Let $p$ and $q$ be two interactions of a term $P$ such that $p$ and $q$ are
  reorderings of each other, then $p\homo_Pq$ and $P/p=P/q$.
\end{proposition}
\begin{proof}
  We first prove that for any interaction $a_1\ldots a_nb∈P$ such that $b∈P$
  we have $a_1\ldots a_nb\homo_Pba_1\ldots a_n$ and
  $P/(a_1\ldots a_nb)=P/(ba_1\ldots a_n)$, by induction on $n$.
  The case $n=0$ is trivial.
  For the case $n≥1$, remark that the hypothesis implies $a_1\indep b$: if
  some location in $a_1$ was less than a location in $b$ then $b$ could only
  occur after $a_1$, which contradicts $b∈P$, and $a_1∈P$ also implies that no
  location in $b$ is less than a location in $a_1$.
  Therefore we have $ba_1∈P$ and $ba_1\homo_Pa_1b$.
  The equality $P/a_1b=P/ba_1$ is a simple check on the transition rules.
  Applying the induction hypothesis on $P/a_1$ yields
  $ba_2\ldots a_n\homo_Pa_2\ldots a_nb$ and
  $P/a_1ba_2\ldots a_n=P/ba_1a_2\ldots a_n$ from which we conclude.
  The case of arbitrary reorderings follows by recurrence on the length of
  $p$ and $q$.
\end{proof}

\begin{definition}
  A pre-trace is a homotopy class of interactions.
  A run is a homotopy class of maximal paths.
  The sets of pre-traces and runs of a term $P$ are written $\Pretraces(P)$
  and $\Runs{P}$ respectively.
  The unique reduct of a term $P$ by a pre-trace $ρ$ is written $P/ρ$.
\end{definition}

Runs are the intended operational semantics: they are complete
executions of a given system, forgetting unimportant interleaving of actions
and remembering only actual ordering constraints.
A pre-trace can be seen as a Mazurkiewicz trace~\cite{dr95:traces} on the
infinite language of transition labels, with the independence relation from
Definition~\ref{def-homotopy}, except that, because of our transition rules,
each label occurs at most once in any interaction.

We now define a form of interactive observation, in the style of
testing equivalences, that takes this notion of homotopy in account.
Standard testing leads to interleaving semantics, so we have to
refine our notion of test, and that is what outcomes are for.
The set $\Scal$ is a semiring in order to represent two ways of combining
results: multiplication is parallel composition of independent results and
addition is combination of results from distinct runs.

\begin{definition}
  The state $s(P)∈\Scal$ of a term $P$ is defined inductively as
  \begin{center}
    $ s(λ) := λ $, \hfil
    $ s(S) := 1 $, \hfil
    $ s(\new{x}P) := s(P) $, \hfil
    $ s(P\para Q) := s(P) \, s(Q) $.
  \end{center}
  The outcome of a term $P$ is $\outcome{P}=\sum_{ρ∈\Runs{P}}s(P/ρ)$.
  Two terms $P$ and $Q$ are observationally equivalent, written $P≃Q$, if
  $\outcome{P\para R}=\outcome{Q\para R}$ for all~$R$.
%
\end{definition}

In other words, the outcome of testing $P$ against $Q$ is the sum of the final
states of all different runs of $P|Q$.
Note that this sum is always finite since we only consider terms without
replication or recursion, hence all terms have finitely many runs.
Classic forms of test intuitively correspond to the case where $\Scal$ is the
set of booleans for the two outcomes success and failure, with operations
defined appropriately.
This particular case is detailed at the end of Section~\ref{sec:consequences}.

\subsection{Decomposition of processes} 
\label{sec:decomp}

In this section, we prove several properties of terms up to observational
equivalence.
The purpose is to decompose arbitrary terms into simpler terms
from which we will be able to easily extract a semantics in order algebras.

\begin{definition}\label{def-causal-order}
  Let $P$ be a term and let $ρ∈\Pretraces(P)$ be a pre-trace of $P$.
  By Proposition~\ref{prop-permutation}, $ρ$ is identified with the set of its
  labels.
  \begin{itemize}
  \item 
    The causal order in $ρ$ is the partial order $≤_ρ$ on labels in $ρ$ such
    that $a≤_ρb$ if $a=b$ or $a$ occurs before $b$ in all interactions in~$ρ$.
  \item 
    The outcome of a pre-trace $ρ$ is defined as $\outcome{ρ}:=s(P/ρ)$.
  \end{itemize}
\end{definition}

This presentation is much simpler to handle than
explicit sets of runs, so this is the one we will mainly use.
Interactions that constitute a given pre-trace are simply the topological
orderings of this partially ordered set of transitions.
Traces in our sense are a further quotient of pre-traces, defined and
studied in Section~\ref{sec-traces}.

\begin{proposition}\label{prop:congruence}
  Observational equivalence is a congruence.
\end{proposition}
\begin{proof}
  Consider a family of equivalent processes $(P_i≃Q_i)_{1≤i≤n}$, and let
  $(α_i)_{1≤i≤n}$ be a family of actions on fresh locations $κ_i$.
  Let $P=\sum_{i=1}^nα_i.P_i$ and $Q=\sum_{i=1}^nα_i.Q_i$, we prove $P≃Q$.
  Let $R$ be an arbitrary process.
  The set $\Runs{P\para R}$ can be split into $n+1$ parts:
  the set $\?R_0$ of runs where no action $α_i$ is triggered and
  the sets $\?R_i$ of runs in which $α_i$ is triggered, for each $i$.
  Then for each run $ρ∈\?R_i$, there is a position $ι$ such that
  $\itlabel{κ_i}{ι}∈ρ$.
  Let $ρ_1$ be the partial run $\set{a}{a∈ρ,a≤_ρ\itlabel{κ_i}{ι}}$, that is
  the minimal run that triggers $α_i$;
  we have $(P\para R)/ρ_1=\new{x}(P_i\para R')$ for some $x$ and $R'$;
  let $ρ_2=ρ∖ρ_1$, so that $ρ_2$ is a run of $P_i\para R'$ and
  $(P\para R)/ρ=\new{x}(P_i\para R')/ρ_2$.
  Let $\?S_i$ be the set of triples $(ρ_1,R',ρ_2)$ for all $ρ∈\?R_i$.
  Obviously $\Runs{P\para R}$ is in bijection with
  $\?R_0\uplus\biguplus_{i=1}^n\?S_i$ and
  \[
    \outcome{P\para R}
    = \sum_{ρ∈\?R_0} s(R/ρ)
    + \sum_{i=1}^n \sum_{(ρ_1,R',ρ_2)∈\?S_i} s((P_i\para R')/ρ_2)
  \]
  Now let $\?L_i=\set{(ρ_1,R')}{∃ρ_2,(ρ_1,R',ρ_2)∈\?S_i}$, and let
  $(ρ_1,R')∈\?L_i$.
  Since $\?R_i$ contains all the runs of $P\para R$ that trigger $α_i$, it
  contains all the runs of $P_i\para R'$ since $P_i\para R'$ can be reached
  from $P\para R$, so we have
  $\set{ρ_2}{(ρ_1,R',ρ_2)∈\?S_i}=\Runs{P_i\para R'}$,
  hence
  \[
    \sum_{(ρ_1,R',ρ_2)∈\?S_i} s((P_i\para R')/ρ_2)
    = \sum_{(ρ_1,R')∈\?L_i} \sum_{ρ_2∈\Runs{P_i\para R'}} s((P_i\para R')/ρ)
    = \sum_{(ρ_1,R')∈\?L_i} \outcome{P_i\para R'}
  \]
  By hypothesis, for all $R'$ we have
  $\outcome{P_i\para R'}=\outcome{Q_i\para R'}$ so
  \[
    \outcome{P\para R}
    = \sum_{ρ∈\?R_0} s(R/ρ)
    + \sum_{i=1}^n \sum_{(ρ_1,R')∈\?L_i} \outcome{Q_i\para R'}
    = \outcome{Q\para R}
  \]
  since the reasoning above equally applies to $Q$.
  Therefore we get $P≃Q$.

  For parallel composition, let $R$ and $S$ be arbitrary terms.
  It is clear that $(P\para R)\para S$ and $P\para(R\para S)$ have the same
  runs and that their reducts by a given run are the same up to the same
  associativity, so for all run $ρ$ we have
  $s(((P\para R)\para S)/ρ)=s((P\para(R \para S))/ρ)$ and therefore
  $\outcome{(P\para R)\para S}=\outcome{P\para(R\para S)}$.
  Similarly we get $\outcome{(Q\para R)\para S}=\outcome{Q\para(R\para S)}$,
  and by hypothesis we have $P≃Q$ so $\outcome{P\para(R\para
  S)}=\outcome{Q\para(R\para S)}$, from which we conclude.

  The equality $\outcome{\new{x}P\para R}=\outcome{\new{x}Q\para R}$ is justified by the
  fact that $\outcome{\new{x}P\para R}$ and $\outcome{P\para R}$ are equal if the name
  $x$ is fresh with respect to $R$.
\end{proof}

\begin{table}
  \begin{tabular}{lll}
    commutativity &
      $ P \para Q  ≃  Q \para P $ &
      $ S + T  ≃  T + S $ \\
    associativity &
      $ (P \para Q) \para R  ≃  P \para (Q \para R) $ &
      $ (S + T) + U  ≃  S + (T + U) $ \\
    neutrality &
      $ P \para 1  ≃  P $ \\[1ex]
    scope commutation &
      $ \new{x}\new{y} P  ≃  \new{y}\new{x} P $ \\
    scope extrusion &
      $ \new{x} (P \para Q)  ≃  P \para \new{x} Q $
      & with $x∉\fn(P)$ \\
    scope neutrality &
      $ \new{x} λ  ≃  λ $ \\[1ex]
    inaction &
      $ \new{u} u^ε(x).P ≃ 1 $ \\
    non-interference &
      $ \new{u}( u(x).P \para \bar{u}(x).Q ) ≃ \new{ux}(P\para Q) $
  \end{tabular}
  \caption{Basic equivalences.}
  \label{table-basic}
\end{table}

\begin{proposition}\label{prop-basic}
  The equivalences of Table~\ref{table-basic} hold.
\end{proposition}
\begin{proof}
  For every equation $A≃B$ in the list except non-interference, 
  it is clear that for all term $T$ we have $\Runs{A\para T}=\Runs{B\para T}$
  and that the reducts by any run $ρ$ differ in the same way.
  Since these rules preserve states, in each case we get
  $\outcome{A\para T}=\outcome{A\para T}$, hence the expected equivalence.
  For the non-interference rule, remark that all runs of
  $\new{u}(\loc{ι}{u}(x).P\para\loc{κ}{\bar{u}}(x).Q)\para R$ contain the
  transition $\itlabel{ι}{κ}$, because of maximality and the fact that $R$
  cannot provide actions on $u$.
  The reduct by this transition is $\new{ux}(P\para Q)\para R$, and
  its runs are those of the original term without $\itlabel{ι}{κ}$, so it has
  the same outcome.
\end{proof}

Thanks to these properties, when considering processes up to observational
equivalence, we can consider parallel composition to be associative and
commutative.
In this case we use the notation $\prod_{i∈I}P_i$ to denote the parallel
composition without interaction of the $P_i$ in any order (assuming only that
$I$ is finite).

In order to study processes up to observational equivalence, we will now
describe some of the structure of the space of equivalence classes.
The first ingredient is to identify an additive structure that represents pure
non-determinism.

\begin{table}
  \begin{tabular}{ll}
  commutative monoid: &
  $ P⊕Q ≃ Q⊕P \qquad
    (P⊕Q)⊕R ≃ P⊕(Q⊕R) \qquad
    P⊕0 ≃ P $ \\[1ex]
  scalar multiplication: &
  $ 0⋅P ≃ 0 \qquad
    1⋅P ≃ P \qquad
    λ_1λ_2⋅P ≃ λ_1⋅(λ_2⋅P) $ \\&
  $ (λ_1+λ_2)⋅P ≃ (λ_1⋅P) ⊕ (λ_2⋅P) \qquad
    λ⋅(P⊕Q) ≃ λ⋅P ⊕ λ⋅Q $ \\[1ex]
  linearity of operators: &
  $ P\para(Q⊕R) ≃ (P\para Q)⊕(P\para R) \qquad
    P\para(λ⋅Q) ≃ λ⋅(P\para Q) $ \\&
  $ \new{x}(P⊕Q) ≃ \new{x}P ⊕ \new{x}Q \qquad
    \new{x}(λ⋅P) ≃ λ⋅\new{x}P $
  \end{tabular}
  \caption{Module laws over processes.}
  \label{table-module}
\end{table}

\begin{proposition}\label{prop-module}
  Let $\Procs_\Scal$ be the set of equivalence classes of processes over the
  semiring of outcomes $\Scal$.
  For all terms $P$ and $Q$ and all outcome $λ$, define
  \begin{syntax}
    \define P⊕Q \case \new{u}((u.P\para u.Q)\para\bar{u}.1)
      \comment{where $u$ is a fresh name,}
    \define λ⋅P \case λ\para P
  \end{syntax}
  Then $(\Procs_\Scal,⊕,0,⋅)$ is a $\Scal$-module,
  parallel compositions are bilinear operators and hiding is linear,
  i.e. the equivalences of Table~\ref{table-module} hold.
\end{proposition}
\begin{proof}
  We first show that, for all terms $P$, $Q$ and $R$,
  $\outcome{(P⊕Q)\para R}=\outcome{P\para R}+\outcome{Q\para R}$.
  Consider $\Runs{(P⊕Q)\para R}
  =\Runs{\new{u}((\loc{ι_1}{u}.P\para\loc{ι_2}{u}.Q)\para\loc{κ}{\bar{u}}.1)\para R}$.
  It is clear that any run contains an interaction of $\bar{u}.1$ with either
  $u.P$ or $u.Q$, since none of these may interact with anything else.
  We can thus write $\Runs{(P⊕Q)\para R}=\?R_1\uplus\?R_2$ where $\?R_1$ is the set
  of runs that contain $(ι_1,κ)$ and $\?R_2$ is the set of runs that contain
  $(ι_2,κ)$.
  The runs in $\?R_1$ are the runs of
  $\new{u}(\loc{ι_1}{u}.P\para\bar{u})\para R$
  and each of these runs has the same outcome in both terms, so
  \[
    \sum_{ρ∈\?R_1} s\bigl(((P⊕Q)\para R)/ρ\bigr)
    = \outcome{\new{u}(\loc{ι_1}{u}.P\para\bar{u})\para R}
    = \outcome{P\para R}
  \]
  by the non-interference rule of Table~\ref{table-basic}.
  By a similar argument, we get the same for $\?R_2$ and $\outcome{Q\para R}$,
  so we finally get
  $\outcome{(P⊕Q)\para R}=\outcome{P\para R}+\outcome{Q\para R}$.
  This equality and the fact that $(\Scal,+,0)$ is a commutative monoid
  implies that $(\Procs_\Scal,⊕,0)$ is a commutative monoid (where
  $0$ is the atomic term with outcome $0$).

  For any terms $P$ and $Q$ and any outcome $λ$, it is clear that
  $\outcome{(λ\para P)\para Q}=λ\outcome{P\para Q}$, since the term $λ$ has no transition and
  contributes $λ$ multiplicatively to all outcomes of the term.
  This directly implies that the operation $λ⋅P$ has all required properties.

  For the bilinearity of compositions, using the equation 
  $\outcome{(P⊕Q)\para R}=\outcome{P\para R}+\outcome{Q\para R}$
  and associativity and commutativity of parallel composition we get
  that parallel composition distributes over $⊕$, and the fact
  that $0$ is absorbing is equivalent to the rule $0⋅P≃0$.
  Linearity of hiding is immediate from the scoping rules and the fact that
  $\outcome{\new{x}P}=\outcome{P}$ always holds.
\end{proof}

Observe that all syntactic constructions induce linear constructions
on equivalence classes, except for the action prefix, which is not linear but
actually affine.
Indeed, for an action $α$, the term $α.0$ is not equivalent to $0$: it will be
neutral in executions that do not trigger $α$, and multiply the outcome by
$0$ (thus annihilating it) in runs that do.
It can be understood as a statement ``I could have performed $α$ but I will not
do it'' so that any run that contradicts this statement has outcome $0$.
The purely linear part of actions is the opposite: the linear action
$\lin{α}.P$ will act as $α.P$ if its environment actually triggers the action,
but will turn to $0$ if it is never activated.

\begin{definition}\label{def-linear-action}
  For all action $α$ and term $P$, the linear action of $α$ on $P$ is
  \begin{syntax}
    \define \lin{α}.P
    \case \new{w}(α.(P\para w.1)\para w.0\para\bar{w}.1)
    \comment{where $w$ is a fresh name.}
  \end{syntax}
  An interaction is said to trigger the linear action if it triggers the
  action $w.1$.
  Terms of the form $α.0$ are called inactions.
\end{definition}

This definition has the expected behaviour because of the maximality of runs.
If $\lin{α}.P$ is in active position, then any run that does not trigger $α$
must instead trigger $w.0$, hence any such run has outcome $0$.
A run in which the term $\lin{α}.P$ does not produce $0$ must activate $α$, so
that $w.1$ acts instead of $w.0$.

In this respect the action $\lin{α}$ is \emph{linear}, in the sense of a
linear resource: it must be used exactly once, otherwise the process must
evolve to~$0$, as stated by the third equation of
Table~\ref{table-linear-actions}.
As proved below, it is also linear as an operator $P\mapsto\lin{α}.P$.
These two features are deeply related: internal choice and outcomes may
commute with the action prefix only if we know for sure that the prefix will
eventually be used.

\begin{table}
  \begin{tabular}{ll@{}}
  Linearity: &
  $ \lin{α}.(P⊕Q) ≃ \lin{α}.P⊕\lin{α}.Q \quad
    \lin{α}.(λ⋅P) ≃ λ⋅\lin{α}.P \quad
    \new{u}\lin{u}^ε(x).P ≃ 0 $ \\[1ex]
  Asynchrony of inactions: &
  $ \lin{α}.(β.0 \para P) ≃ β.0 \para \lin{α}.P $
   \qquad if the subject of $β$ is not bound by $α$ \\[1ex]
  Composition of inactions: &
  $ \sum_{i∈I} α_i.0 \para \sum_{i∈J} α_i.0 ≃
    \begin{array}[t]{@{}l}
      0 \quad\text{if } α_i=\bar{α}_j \text{ for some } i∈I, j∈J \\
      \sum_{i∈I\cup J} α_i.0 \quad\text{otherwise}
    \end{array} $ \\[1ex] &
  $ α.0 + α.0 ≃ α.0 $
  \end{tabular}
  \caption{Laws of linear actions and inactions.}
  \label{table-linear-actions}
\end{table}

\begin{proposition}\label{prop-affine-action}
  For all families of actions $(α_i)_{i∈I}$ and processes $(P_i)_{i∈I}$,
  \[
    \sum_{i∈I} α_i.P_i ≃
    \bigoplus_{i∈I} \lin{α}_i.P_i
    ⊕
    \sum_{i∈I} α_i.0 .
  \]
  The function $P\mapsto\lin{α}.P$ is linear and the equivalences of
  Table~\ref{table-linear-actions} hold.
\end{proposition}
\begin{proof}
  For linearity, we use the fact that $\outcome{\lin{α}.P\para Q}$ is the sum
  of the $s((\lin{α}.P\para Q)/ρ)$ for the runs $ρ$ that actually trigger $α$
  (and the witness action $w.1$).
  If $P=λ\para P'$ for some $λ∈\Scal$, these runs are the same in
  $\lin{α}.(λ\para P')\para Q$ and $\lin{α}.P'\para Q$, but the outcomes are
  multiplied by $λ$ in the first case, so $\outcome{\lin{α}.(λ\para P')\para
  Q}=λ⋅\outcome{\lin{α}.P'\para Q}$ and $\lin{α}.(λ\para
  P')≃λ\para\lin{α}.P'$.
  If $P=P_1⊕P_2$, the choice is eventually active in all relevant runs, so
  each of these runs triggers either $P_1$ or $P_2$.
  We can thus establish a bijection between $\Runs{\lin{α}.(P_1⊕P_2)\para Q}$
  and the disjoint union of $\Runs{\lin{α}.P_1\para Q}$ and
  $\Runs{\lin{α}.P_2\para Q}$.
  Since outcomes are preserved by this bijection, we finally get
  $\outcome{\lin{α}.(P_1⊕P_2)\para Q}=
  \outcome{\lin{α}.P_1\para Q}+\outcome{\lin{α}.P_2\para Q}$
  and $(P_1⊕P_2)\para Q≃(P_1\para Q)⊕(P_2\para Q)$.

  The equivalence $\new{u}u^ε(x).P≃0$ can be deduced from previous equations:
  \begin{align*}
    \new{u}\lin{u}^ε(x).P
    &= \new{uw}(u^ε(x).(P\para w.1) \para (w.0 \para \bar{w}.1)) \\
    &≃ \new{w}(\new{u}u^ε(x).(P\para w.1) \para (w.0 \para \bar{w}.1)) \\
    &≃ \new{w}(1 \para (w.0 \para \bar{w}.1))
    ≃ \new{w}(w.0 \para \bar{w}.1)
    ≃ \new{w}(0 \para 1)
    ≃ 0
  \end{align*}

  For the decomposition, let $f$ and $g$ be the functions from $\Procs_\Scal$
  to $\Scal$ such that $f(Q)=\outcome{ (\sum_{i∈I} α_i.P_i) \para Q }$ and
  $g(Q)=\outcome{ (\bigoplus_{i∈I} \lin{α}_i.P_i
    ⊕ \sum_{i∈I} α_i.0 ) \para Q }=\sum_{i∈I}g_i(Q)+g_0(Q)$,
  we prove $f=g$.
  By previous remarks we have
  $g(Q)=\sum_{i∈I}\outcome{\lin{α}_i.P_i\para Q}
    +\outcome{\sum_{i∈I}α_i.0\para Q}$.
  Given a term $Q$, $\Runs{ (\sum_{i∈I} α_i.P_i) \para Q }$
  decomposes into $\?R_0$ for the runs that trigger none of the $α_i$
  and a $\?R_i$ for all runs that trigger $α_i$, for each $i$.
  Clearly, $\?R_0$ contains the runs of $\sum_{i∈I}α_i.0\para Q$ that do
  not trigger any $α_i$, and all other runs of this term have outcome $0$, so
  the sum of the outcomes of runs in $\?R_0$ is $g_0(Q)$.
  For each $i$, the runs of $\?R_i$ are in bijection with runs of
  $\lin{α}_i.P_i\para Q$ that trigger $α_i$ and they have the same outcomes,
  and all other runs of this term have outcome $0$, so the sum out the
  outcomes of these runs is $g_i(Q)$.
  As a consequence, we get the expected decomposition
  $f=g_0+\sum_{i∈I}g_i$.

  For the equivalence $\lin{α}.(β.0\para P)≃β.0\para\lin{α}.P$,
  assuming the subject of $β$ is not the bound name of action $α$,
  let $Q$ be an arbitrary term and consider $\Runs{\lin{α}.(β.0\para P)\para Q}$.
  Any run that does not trigger $\lin{α}$ or that triggers both $\lin{α}$ and
  $β$ has outcome $0$, so the only relevant runs are those that trigger
  $\lin{α}$ but not $β$.
  Clearly these runs are the same as the runs of $(β.0\para\lin{α}.P)\para Q$
  that trigger $\lin{α}$ and not $β$, and they have the same outcomes.

  For the composition of inactions, the relevant runs of a term
  $(\sum_{i∈I}α_i.0\para\sum_{i∈J}α.0)\para P$ are those that do not
  trigger any of the $α_i$, so the number of occurrences of each $α_i$ does
  not matter, and the fact that they are in a branching or in parallel does
  not matter either, as long as the branchings cannot interact.
  The only special case is when there are $i∈I$ and $j∈J$ such that
  $α_i=\bar{α}_j$, then each run must trigger one branching or the
  other, if nothing else by letting $α_i$ and $α_j$ interact.
  As a consequence, all runs of this term have outcome $0$, so the composition
  of the two branchings is indistinguishable from~$0$.
\end{proof}


\begin{definition}\label{def-simple-term}
  A term is \emph{simple} if it is generated by the following grammar
  \begin{syntax}
  \define[simple term] P,Q
    \case 1 ,\; N ,\; \lin{α}.P ,\; (P\para Q) ,\; \new{x}P
  \define[inaction set] N
    \case \sum_{i∈I}α_i.0
  \end{syntax}
  A pre-trace $ρ∈\Pretraces(P)$ is exhaustive if it triggers all linear
  actions and no inaction, and no sub-term of $P/ρ$ has the form $Q\para R$ with
  $Q$ containing some $α.0$ and $R$ containing $\bar{α}.0$.
  The set of such pre-traces is written $\Pretraces_e(P)$.
\end{definition}


Simple terms have the property that the outcome of any run is either $1$ or
$0$.
More precisely, it is easy to see that the outcome of a run is $1$ if and only
if it triggers all linear actions and no inaction.
The notion of exhaustive pre-trace is the correct extension of this notion to
pre-traces, indeed every run of a simple term $P\para Q$ with outcome $1$ is
made of an exhaustive pre-trace of $P$ and an exhaustive pre-trace of $Q$.
The condition on $P/ρ$ simply rules out interactions of $P$ that
lead to a term $P'$ where there are dual inactions that may interact, since
that would imply $P'≃0$, as a generalization of the equation
$α.0\para\bar{α}.0≃0$.
Observe that, by the decomposition of Proposition~\ref{prop-affine-action} and
the linearity of all constructions of simple terms, we immediately prove that
every term is equivalent to a linear combination of simple terms.
As a consequence, two terms $P$ and $Q$ are equivalent if and only if for all
\emph{simple} term $R$, $\outcome{P\para R}=\outcome{Q\para R}$.

\subsection{An order algebraic model} 
\label{sec-traces}

Thanks to the decomposition into simple terms, we are now ready to describe
our order algebraic semantics.
Following the initial intuition, we define a web whose points are action
occurrences, with a group action that permutes actions of the same name and
polarity while making sure that bound names are properly updated.
We need an extra bit of information to represent inactions, and these will be
represented as extra actions (somehow ``potential'' actions) with particular
treatment.

\begin{definition}\label{def:pii-arena}
  The set $C$ of \emph{abstract channels} is defined as
  $ C := \Names × (\Pola × \Nat)^* $.
  We write $u⋅ε_1n_1\cdotsε_kn_k$
  instead of $(u,((ε_1,n_1),\ldots,(ε_k,n_k)))$.

  The arena $E$ for the \piI-calculus is such that
  $ \web{E} = C × \Pola × (\Nat\cup\{\bot,\top\}) $
  and $\pgroup[E]$ is generated by permutations of the form
  $(x,ε,σ)∈C×\Pola×\Perm{\Nat}$, acting as
  \[
    (x,ε,σ)(y) =
    \begin{cases}
      x⋅εσ(n)⋅z &\text{if } y = x⋅εn⋅z
        \text{ for some } n∈\Nat \text{ and } z∈(\Pola×\Nat)^* \\
      y &\text{otherwise}
    \end{cases}
  \]
\end{definition}

Abstract channels represent names in a way that allows us to
avoid any renamings.
Intuitively, $u$ (that is $(u,())$) represents the free name $u$ itself,
$u⋅εn$ represents the bound name $x$ in the action $\loc{n}{u}^ε(x)$, then
$u⋅εn⋅ε'n'$ represents the bound name $y$ in $\loc{n'}{x}^{ε'}(y)$, and
so on.
So an abstract channel is the path to find a given name, free or bound.
In a sense, this notion is an analogous for names in \piI-terms of De Bruijn
indices.

We can assume, without loss of generality, that all names in processes we
use respect this intuition, so that we can mention any name without ambiguity
and with no need of renaming.
Under this hypothesis, given a term $P$ and a pre-trace $ρ∈\Pretraces(P)$,
the term $P/ρ$ is uniquely defined, \emph{not} up to renaming.
Note however that $P/ρ$ does not respect the intuition on free names if
bound names were revealed, \ie\ if $ρ$ contains a visible action.
With this discipline on names, we can assume without loss of generality that
the set of free names $\Names$ is \emph{finite}.

Points in the web $\web{E}$ are of two kinds.
The first kind is $x⋅εn$ for the occurrence of polarity $ε$ at location $n$ of
the name $x$.
The second kind is $x⋅ε\bot$ or $x⋅ε\top$ for the inaction of polarity $ε$
with name $x$; the use of $\bot$ and $\top$ is a tool used to encode the
behaviour of inactions, with the convention that $x⋅ε\bot<x⋅ε\top$ if $x^ε.0$
is present, and the points are incomparable otherwise.

A permutation $(x,ε,σ)$ permutes the locations of the actions of polarity $ε$
of the name $x$ according to $σ:\Nat\to\Nat$.
By definition, the $n$-th occurrence of polarity $ε$ of $x$, namely $x⋅εn$, is
renamed into $x⋅εσ(n)$, the $m$-th occurrence of polarity $η$ of the name
bound by it, namely $x⋅εn⋅ηm$, gets renamed as $x⋅εσ(n)⋅ηm$, \ie\ its location
is unchanged but its name is changed to reflect the change of its binder, and
so on for other bound names.
The inactions at $x⋅ε$ are unchanged since $x$ is unchanged,
but those on names bound by $x$ are renamed accordingly.
A more explicit (but equivalent) construction of the permutation group
consists in setting $\pgroup[E]:=\Perm{\Nat}^C$ and defining the action of
$σ∈\pgroup[E]$~as
\[
  σ (u⋅ε_1n_1\cdotsε_kn_k) :=
  u⋅ε_1σ(u)(n_1)⋅ε_2σ(u⋅ε_1n_1)(n_2)\cdots
    ε_kσ(u⋅ε_1n_1\cdotsε_{k-1}n_{k-1})(n_k)
\]
except if $n_k∈\{\bot,\top\}$ in which case the last pair remains as
$ε_kn_k$.

%

\begin{definition}\label{def-trace}
  A trace is a play $t$ on the web $E$ such that
  \begin{itemize}
  \item for all $x⋅εn⋅ε'n'∈\web{t}$,
    $x⋅εn∈\web{t}$ and $x⋅εn<_tx⋅εn⋅ε'n'$,
  \item for all $x∈\Names$ and all $x=y⋅ε'n∈\web{t}$,
    $x⋅ε\bot$ and $x⋅ε\top$ are in $\web{t}$,
    and for all $y∈\web{t}∖\{x⋅ε\bot,x⋅ε\top\}$,
    $x⋅ε\bot$ and $x⋅ε\top$ are incomparable with $y$.
  \end{itemize}
\end{definition}

The first condition is a kind of ``justification'' condition in the style of
game semantics~\cite{ho00:pcf}.
It means that for an action $x⋅εn∈\web{t}$, if the subject $x$ is a bound
name, then its binder (the action also named $x$) is also in $\web{t}$ and it
is inferior in the scheduling order, \ie\ it was revealed earlier.
The second condition means that inactions information must be present for each
known name and that inactions are not involved in scheduling.

\begin{definition}
  Let $P$ be a simple term and let $ρ$ be an exhaustive pre-trace of $P$.
  The trace induced by $ρ$ is the trace $ρ^*$ such that
  \begin{itemize}
  \item $\web{ρ^*}=\set{x⋅εn}{\vtlabel{x^ε}{n}∈ρ}\cup N_ρ$,
    where $N_ρ$ contains $x⋅ε\bot$ and $x⋅ε\top$ for all polarity $ε$ and all
    name $x$ such that $x∈\Names$ or $x=y⋅εn$ for some $\vtlabel{y^ε}{n}∈ρ$,
  \item $≤$ is the causal order (as of Definition~\ref{def-causal-order})
    restricted to visible transitions,
    augmented with $x⋅ε\bot<x⋅ε\top$ for each $x^ε.0$ that occurs in $P/ρ$.
  \end{itemize}
\end{definition}

Note that the justification condition is satisfied by $ρ^*$, because in the
\piI-calculus the action prefixes are synchronous:
in an action $u(x).P$, the action $u(x)$ that binds $x$ is automatically a
prefix of all actions on $x$.
However, synchrony is not necessary for this property to hold, the fact that
the name is bound is the important point: even if internal transitions can
occur on a bound name, visible transitions are possible only after the name
has been revealed by the action it is bound to.

\begin{proposition}\label{prop:translation}
  To each simple term $P$, associate the function
  $\translate{P}:\Plays{E}\to\Scal$ such that for all $t∈\Plays{E}$,
  $\translate{P}(t):=\sharp\set{ρ∈\Pretraces_e(P)}{ρ^*=t}$.
  This function clearly has finite support, so
  $\translate{P}∈\PreOrdAlg{E}$.
  Let $u\mapsto\bar u$ be the linear map over $\PreOrdAlg{E}$ that 
  inverts polarities and exchanges $\bot$ and $\top$.
  Then for all simple terms $P,Q$,
  $\outcome{P\para Q}=\outcome{\translate{P}\psync\overline{\translate{Q}}}$.
\end{proposition}
\begin{proof}
  By construction, if $P$ and $Q$ are simple terms, then so is $P\para Q$, so
  all its runs have outcome $0$ or $1$, thus $\outcome{P\para Q}$ is the
  number of non-zero runs of $P\para Q$.
  Every run $ρ∈\Runs{P\para Q}$ can be uniquely decomposed as a pre-trace
  $ρ_1∈\Pretraces(P)$ and a pre-trace $ρ_2∈\Pretraces(Q)$.
  Moreover, by definition of exhaustive pre-traces, if the outcome of $ρ$ is
  $1$ then $ρ_1$ and $ρ_2$ are exhaustive pre-traces.

  Now let $ρ_1$ and $ρ_2$ be any exhaustive pre-traces of $P$, we want to
  compute how many runs with outcome $1$ they generate.
  A run $ρ∈\Runs{P\para Q}$ projects to $ρ_1$ and $ρ_2$ if and only if it
  establishes a bijection from visible actions of $ρ_1$ to dual visible
  actions of $ρ_2$, such that scheduling constraints are respected and
  no opposite inactions exist between $ρ_1$ and $ρ_2$.
  Formulated in traces, this means a bijection $φ:\web{ρ_1^*}\to\web{ρ_2^*}$
  such that:
  \begin{itemize}
  \item 
    For all $a=x⋅εn∈\web{ρ_1^*}$,
    $φ(a)=y⋅\negεm$ for some $y$ and $m$
    (\ie\ actions of opposite polarities are matched),
    and if $x∈\web{ρ_1^*}$ then $y∈\web{ρ_2^*}$ and $φ(x)=y$.
    This means that an action on a bound name must be matched with an
    action on another bound name and that these names are revealed by actions
    that were matched together (this is a typical property of the
    \piI-calculus).
  \item 
    The union of the orders $φ(≤_{ρ_1^*})$ and $≤_{ρ_2^*}$ is acyclic,
    which means that $φ$ respects prefixing constraints
    so that we get an actual execution path.
  \end{itemize}
  Such a bijection $φ$ establishes an identification between names revealed in
  the interactions $ρ_1$ and $ρ_2$, and the last thing to check is that under
  this bijection, there are no dual inactions between $ρ_1$ and $ρ_2$.
  By construction, that there are such inactions if and only if
  for some name $x$ in $\Names$ or $\web{ρ_1^*}$ and polarity $ε$ we have
  $x⋅ε\bot<_{ρ_1^*}x⋅ε\top$ and $φ(x)⋅\negε\top<_{ρ_2^*}x⋅\negε\bot$, which
  exactly corresponds to a cycle in the union
  $φ(≤_{ρ_1^*})\cup\overline{≤_{ρ_2^*}}$.

  It is routine to check that bijections that satisfy the above conditions are
  exactly the bijections induced by elements of $\pgroup[E]$ such that
  $φ(ρ_1^*)\sync\overline{ρ_2^*}=1$: the structure of $\pgroup[E]$ is made to
  ensure that the justification condition is satisfied, and the rest is
  ordering conditions.
  As a consequence, the number we seek is exactly
  $ρ_1^*\psync\overline{ρ_2^*}$.
  By summing this on all pairs of exhaustive pre-traces of $P$ and $Q$, we
  finally get
  $\outcome{P\para Q}=\outcome{\translate{P}\psync\overline{\translate{Q}}}$.
\end{proof}

The translation function $P\mapsto\translate{P}$ defined above applies to
simple terms, but using the results of Section~\ref{sec:decomp} we can extend
it to all terms by linear combinations.
The decomposition of terms as linear combinations of simple terms is not
unique syntactically, however all decompositions are observationally
equivalent by definition, and it is easy to check that the traces induced by
all possible translations of a given term are the same, so the translation
is actually a function from terms to vectors in $\PreOrdAlg{E}$.
The space of linear combinations of plays $\PreOrdAlg{E}$ is larger than the
set of translations of terms, so by Proposition~\ref{prop:translation} if
translations of two terms $P$ and $Q$ are observationally equivalent in the
sense of order algebras then these terms are equivalent in the sense of
quantitative testing.
Hence our final theorem:
\begin{theorem}
  Two terms of the \piI-calculus are observationally equivalent for
  quantitative testing in a semiring $\Scal$ if and only if their translations
  in $\OrdAlg[\Scal]{E}$ are equal.
\end{theorem}

\subsection{Consequences} 
\label{sec:consequences}

The first consequence of this model is that Theorem~\ref{thm:basis} provides a
basis for the set of processes in two particular cases:
\begin{itemize}
\item If $\Scal$ is idempotent, then each term is equivalent to a linear
  combination of totally ordered traces.
  It is the case when $\Scal$ represents standard may or must testing.
  Then we lose the ``quantitative'' aspect since multiplicities are ignored,
  and we fall back to standard semantics as a special case.
  We get full abstraction in this case by showing that any base play can be
  implemented as a term of the calculus~\cite{bef08:apc}.
\item If $\Scal$ is a regular ring, terms are combinations of weakly totally
  ordered traces.
  We can get full abstraction again if we slightly extend the calculus to
  allow parallel composition without interaction~\cite{bef09:qt}, this is
  needed only for the case of traces that contain concurrent dual actions.
  Actually the only needed feature is a multiple prefix
  $\{α_1,\ldots,α_k\}.P$, which is enough to represent weakly ordered traces
  as terms.
  Simpler modifications of the calculus could lead to full abstraction, for
  instance by imposing a more structured naming discipline.
\end{itemize}

Although we will not write the proofs here in full detail, the interpretation
of processes is compositional, and we can use the constructs of
Section~\ref{sec:logic} to represent syntactic constructs as operators on
order algebras.
Define the arena $Ch$ of channel ends as
$\web{Ch}=(\Nat×\Pola)^*×(\Nat\cup I)$ with $I=\{\bot,\top\}$,
with permutations of the same kind as in $E$, then the definition of $E$ from
Definition~\ref{def:pii-arena} reformulates as
\[
  E = (\Names × \Pola) \indexing Ch
  \qquad\text{and}\qquad
  Ch = I + \sharp(\{\ast\}+\Pola\indexing Ch)
\]
up to a simple isomorphism.
These equations mean that a process appears as a family of
channel ends indexed by free names and polarities, and that a channel end
contains inaction information (the $I$ part) and an arbitrary number of
interchangeable occurrences (the $\ast$) each associated with a new
channel end per polarity (the $\Pola\indexing Ch$).

The explicit mention of the $\sharp$ operator for the action occurrences
allows us to use the $γ$ and $δ$ operators from
Definition~\ref{def:gamma-delta} as a systematic way of treating the inherent
non-determinism in the multiple occurrences of each name.
We can thus define parallel composition of vectors $p\para q$ in
the order algebra as follows:
\begin{itemize}
\item For each channel end $x⋅ε$ in $P$, apply
  $δ^2:\OrdAlg{\sharp Oc}\to\OrdAlg{\sharp Oc+\sharp Oc}$, where
  $Oc=\{\ast\}+\Pola\indexing Ch$ is the arena for an action occurrence.
  This splits the occurrences of $x⋅ε$ into those that will interact with $Q$
  and those that will not.
  Extend this to $Ch$ by keeping the inaction part unchanged, and apply the
  same operator independently to each channel end, giving an
  operator $δ':\OrdAlg{E}\to\OrdAlg{E+(\Names×\Pola)\indexing\sharp Oc}$.
  The $\sharp Oc$ part contains actions that will \emph{not} interact.
\item Do the same for $Q$, and compose the result with the involution
  $u\mapsto\bar u$ from Proposition~\ref{prop:translation}.
\item Partially synchronize $δ'(p)$ and $\overline{δ'(q)}$ on the $E$
  part, to represent the actual interaction for the occurrences that must
  interact, which yields a vector
  $u∈\OrdAlg{E+(\Names×\Pola)\indexing(\sharp Oc+\sharp Oc)}$.
  This partial synchronisation handles the conditions on inactions the same
  way as in Proposition~\ref{prop:translation}.
\item In the result, for each channel end $x⋅ε$ in the $E$ part, forget the
  actions on $x⋅ε$ since they have interacted, then normalise the inaction
  part by mapping any $y⋅ε\top<y⋅ε\bot$ to the reverse order (this is a linear
  operator since it acts on plays) and inverting again the remaining part of
  $Q$ to get back the original polarities on visible actions.
  Call $n:\OrdAlg{E+(\Names×\Pola)\indexing(\sharp Oc+\sharp Oc)}\to
  \OrdAlg{(\Names×\Pola)\indexing(I+\sharp Oc+\sharp Oc)}$ this operator.
\item Finally, contract the action occurrences on each channel end in the
  result with the operator
  $γ^2:\OrdAlg{\sharp Oc+\sharp Oc}\to\OrdAlg{\sharp Oc}$ applied on each
  channel end in $\Names×\Pola$, which defines an operator
  $γ':\OrdAlg{(\Names×\Pola)\indexing(I+\sharp Oc+\sharp Oc)}
    \to\OrdAlg{E}$.
\end{itemize}
With this definitions, we finally get
$p\para q:=γ'(n(δ'(p)\psync[E]\overline{δ'(q)}))$.

The other operators are easy to define.
An outcome $λ$ is translated as $λ.\emptyset$, where $\emptyset$ is the empty
run.
Branchings are decomposed as in Proposition~\ref{prop-affine-action}, and the
linear action is a linear operator that consists in introducing in each play
an extra point for the new action, minimal for the scheduling order.
Hiding a name $x$ consists in mapping to $0$ all plays that contain an action
on $x$ and forgetting the inaction information on $x$.

\medskip

By choosing appropriate structures for $\Scal$, we
can recover the standard may and must testing~\cite{dnh84:testing}.
In both cases we have $\Scal=\implem{0,1,ω}$, where $ω$ represents success.
Table~\ref{table-may-must} shows the rules for addition and multiplication for
may and must.
Using this definition it is clear that $P$ and $Q$ are equivalent for may or
must testing if and only if, for all $R$, $\outcome{P\para R}=ω$ if and only if
$\outcome{Q\para R}=ω$.
Taking for $\Scal$ the minimal semiring $\{0,1\}$ with $1+1=1$ gives the
framework studied by the author in a previous work~\cite{bef08:apc}, which
also leads to must testing semantics.
In these semirings, all elements are idempotent for addition, so by
Theorem~\ref{thm:basis} the model we get is actually interleaving.

\begin{table}
  \centering
  \begin{tabular}{c@{\qquad}c@{\qquad}c}
    may and must &
    may testing &
    must testing \\[.5ex]
    $ \begin{array}{c|ccc}
        ⋅ & 0 & 1 & ω \\ \hline
        0 & 0 & 0 & 0 \\
        1 & 0 & 1 & ω \\
        ω & 0 & ω & ω
      \end{array} $ &
    $ \begin{array}{c|ccc}
        + & 0 & 1 & ω \\ \hline
        0 & 0 & 1 & ω \\
        1 & 1 & 1 & ω \\
        ω & ω & ω & ω
      \end{array} $ &
    $ \begin{array}{c|ccc}
        + & 0 & 1 & ω \\ \hline
        0 & 0 & 1 & ω \\
        1 & 1 & 1 & 1 \\
        ω & ω & 1 & ω
      \end{array} $
  \end{tabular}
  \smallskip
  \caption{Observation semirings for may and must testing.}
  \label{table-may-must}
\end{table}


\bibliographystyle{ebstyle}
\bibliography{oa}

\end{document}